\documentclass{article}

\PassOptionsToPackage{numbers, compress}{natbib}
\usepackage[preprint]{neurips_2026}

\usepackage[utf8]{inputenc} 
\usepackage[T1]{fontenc}    
\usepackage{hyperref}       
\usepackage{url}            
\usepackage{booktabs}       
\usepackage{amsfonts}       
\usepackage{nicefrac}       
\usepackage{microtype}      
\usepackage{xcolor}         
\usepackage{footnote}
\usepackage{amsmath}
\usepackage{amssymb}
\usepackage{mathtools}
\usepackage{amsthm}
\usepackage[inline]{enumitem}
\usepackage{setspace}
\usepackage{enumitem}
\usepackage{balance}
\usepackage{url}
\usepackage{multirow}
\usepackage{multicol}
\usepackage{array}
\usepackage{bbm}
\usepackage{tikz}
\usepackage[capitalize,noabbrev]{cleveref}
\usepackage{algorithm}
\usepackage{algorithmic}
\usepackage{wrapfig}
\usepackage[most]{tcolorbox}
\theoremstyle{plain}
\newtheorem{theorem}{Theorem}

\newtheorem{lemma}[theorem]{Lemma}

\theoremstyle{definition}
\newtheorem{definition}{Definition}

\theoremstyle{remark}

\newtheorem{innercustomgeneric}{\customgenericname}
\providecommand{\customgenericname}{}

\newcommand{\newcustomtheorem}[2]{%
  \newenvironment{#1}[1]
  {%
   \renewcommand\customgenericname{#2}%
   \renewcommand\theinnercustomgeneric{##1}%
   \innercustomgeneric
  }
  {\endinnercustomgeneric}
}

\newcustomtheorem{customthm}{Theorem}
\newcustomtheorem{customlemma}{Lemma}
\newcommand{\method}{\texttt{CAPER}}

\usepackage[textsize=tiny]{todonotes}

\newcommand{\circledchar}[2][gray!70]{%
    \tikz[baseline=(char.base)]{
        \node[shape=circle,draw=#1,fill=#1,text=white,inner sep=0.75pt,scale=0.85] (char) {#2};
    }
    }
\title{CAPER: Clause-Aligned Process Supervision for Text-to-SQL}

\author{%
\begin{minipage}{0.99\textwidth}
\centering
Lujie Ban$^{1}$,
Jiasheng Shi$^{1}$,
Jinyang Li$^{2}$,
Xiaolin Han$^{3}$,
Tsz Nam Chan$^{4}$,
Chenhao Ma$^{1*}$ \\
\vspace{0.25em}
$^{1}$The Chinese University of Hong Kong, Shenzhen
\quad
$^{2}$The University of Hong Kong
\quad
$^{3}$Northwestern Polytechnical University
\quad
$^{4}$Shenzhen University \\
\vspace{0.25em}
$^{1}$\texttt{lujieban@link.cuhk.edu.cn, \{shijiasheng,machenhao\}@cuhk.edu.cn} \\
$^{2}$\texttt{jl0725@connect.hku.hk}
\quad
$^{3}$\texttt{xiaolinh@nwpu.edu.cn}
\quad
$^{4}$\texttt{edisonchan@szu.edu.cn} \\
\vspace{0.25em}
\end{minipage}
}

\begin{document}

\maketitle

\begin{abstract}
Text-to-SQL systems are typically evaluated by query-level execution correctness, but this terminal signal provides little guidance about which intermediate SQL decision caused success or failure. Token-level dense supervision is also ill-suited: SQL tokens do not align with complete semantic decisions, can penalize execution-equivalent queries, and are difficult to label reliably at scale. We therefore propose {\method}, which automatically derives clause-level supervision via counterfactual intervention on the SQL abstract syntax tree, enabling root-cause error localization for reward modeling; the resulting data is used to train \texttt{\method-9B}, a lightweight Clause-PRM that provides clause-boundary feedback for policy optimization and candidate verification. Experiments on BIRD and Spider show that clause-aligned supervision not only improves execution accuracy, achieving up to a 15.3\% relative EX improvement over GPT-5.4, but also strengthens failure-localization capability, reaching 84.53\% accuracy and 90.60\% MRR on held-out failures. Our project page is at \hyperlink{link}{https://github.com/banrichard/RL-NL2SQL}.
\end{abstract}  

\section{Introduction}
Text-to-SQL aims to translate natural-language questions into executable SQL queries, providing a direct interface between non-expert users and structured databases \cite{liu2025text2sql-llm-era,hong2025next}.
In real-world data workflows, it supports natural-language analytics~\cite{luo2020deepeye,luo2020steerable}, business-intelligence question answering~\cite{maddela2025starqa,sheinin2018quest}, enterprise database operations~\cite{chen2025adept,xu2025abacus,zeng2020photon}, and interactive SQL troubleshooting~\cite{liswe,huobird}.
Large language models (LLMs) have further expanded the design space through schema linking~\cite{wang2025linkalign,gan-etal-2023-appraising}, database retrieval~\cite{kothyari2023crush4sql,wu2025ucs}, and generation correction~\cite{xu2025ts,qu2025share}. Despite this progress, Text-to-SQL remains a high-precision structured generation problem: a generated query must be syntactically valid, schema-grounded, and semantically equivalent to the intended database operation, and a single wrong join key, aggregation, or predicate can still lead to the wrong execution result.

One way to improve Text-to-SQL is to scale supervised fine-tuning (SFT) \cite{li2025omnisql} with larger or richer demonstration corpora. But this shifts the bottleneck to data construction: executable Text-to-SQL annotations are expensive to produce, require SQL expertise, and remain costly to scale~\citep{liswe,huobird,li2023bird,yu-etal-2018-spider,lei2024spider2,wolfson-etal-2022-weakly,zhang-etal-2025-exesql}. In the LLM era, strong SFT pipelines often also depend on richer supervision such as reasoning traces \citep{he-etal-2025-star-sql}. Query-level execution outcomes offer a more scalable source of supervision because they are already used to judge whether a generated SQL query retrieves the correct answer \citep{zhong2017seq2sql,ma2025sql,yao2025arctic_text2sql_r1,weng-etal-2025-graph}. However, an execution outcome is only a terminal label: it says whether the full query succeeds, but not which intermediate SQL decision made it succeed or fail.

This creates a credit assignment problem for process supervision. As illustrated in \Cref{fig:intro_case}, a predicted SQL query can be largely correct yet fail because of one decisive semantic error, such as an incorrect relational condition or join path. Execution-level feedback treats this near-miss the same as a completely wrong query: once the final execution result is incorrect, the signal provides little guidance about which intermediate SQL decision should be reinforced, repaired, or discouraged. In other words, query-level outcomes provide supervision without root-cause localization: they tell us that a query is wrong, but not which SQL decision caused the failure.
\begin{figure}[t]
    \centering
    \includegraphics[width=\linewidth]{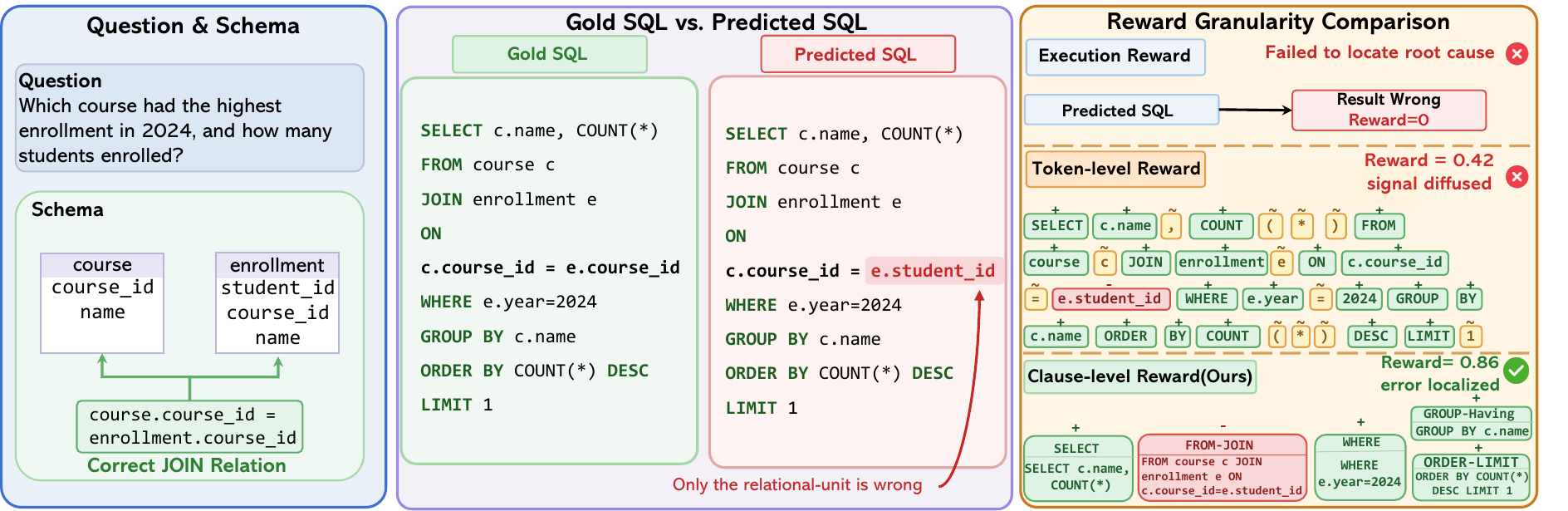}
    \caption{A near-miss Text-to-SQL case illustrating the granularity gap between execution-level outcomes, token-level feedback, and clause-level SQL decisions.}
    \label{fig:intro_case}
    \vspace{-1.6em}
\end{figure}

A na\"ive remedy is to move from execution-level feedback to denser token-level feedback. However, token-level supervision is often too fine-grained for SQL semantics. Individual tokens such as aliases, commas, operators, or column names rarely form complete decisions on their own, while many SQL decisions are expressed by multi-token structures such as selected column sets, predicates, aggregations, or join paths. As a result, token-level feedback can diffuse credit across low-level symbols and require unreliable fine-grained labels. 

These limitations suggest that the right unit of process supervision should lie between the two extremes: denser than execution-level feedback but more semantically grounded than token-level feedback. We therefore focus on clause-level SQL units, which naturally group tokens into meaningful intermediate decisions such as \texttt{SELECT}, \texttt{FROM}/\texttt{JOIN}, \texttt{WHERE}, and \texttt{GROUP BY}/\texttt{HAVING}. Clause-level supervision preserves the scalability of outcome-derived feedback while assigning credit at a granularity closer to how SQL errors affect execution. The key missing step is to localize the root cause of failure at this semantic granularity before converting it into process supervision. This raises question~\circledchar{Q}: \emph{how can we derive clause-aligned supervision for Text-to-SQL from query-level execution outcomes, without requiring costly manual annotations while preserving execution correctness as the task objective?}

\noindent{\textbf{Our Solution.}} To solve question \circledchar{Q}, we propose \circledchar[black]{1} {\method}, an automatic framework that uses counterfactual intervention to localize root-cause SQL decisions and construct clause-level preference supervision, and \circledchar[black]{2} {\method}\texttt{-9B}, a lightweight clause-level process reward model trained for process supervision. Methodologically, {\method} leverages \textbf{$\clubsuit$~counterfactual intervention}, systematically perturbing the abstract syntax tree (AST) to identify the decisive error step responsible for failure. To enhance data diversity and ensure annotation precision, it further adopts \textbf{$\spadesuit$ fault injection}, synthetically generating failure instances by perturbing successful ASTs through targeted corruption.

Further, we fine-tune \texttt{Qwen3.5-9B} on the annotated dataset to obtain {\method}\texttt{-9B} and deploy it as a process reward model in reinforcement learning (RL) for Text-to-SQL, providing fine-grained clause-level rewards for intermediate SQL decisions. A \textbf{multi-granular reward} is designed to supervise RL training, emphasizing both \emph{clause-level} and \emph{execution-level} feedback.
During policy optimization, this multi-granular reward design provides the policy model with both clause-level process feedback and execution-level outcome feedback, yielding denser supervision while preserving final execution accuracy as the task objective.

In summary, our contributions are:
\begin{enumerate}[leftmargin=*]
    \renewcommand{\labelenumi}{\circledchar[black]{\arabic{enumi}}}
    \item \textbf{Automated Pipeline.} We propose a counterfactual SQL-AST pipeline for root-cause localization and clause-level credit assignment without manual annotation, which yields over $90{,}000$ clause-annotated preference tuples across three datasets.
    \item \textbf{Clause-level PRM.} We develop {\method}\texttt{-9B}, a lightweight clause-level process reward model for Text-to-SQL RL optimization through clause-boundary dense rewards, combining intermediate process feedback with final execution correctness.
    \item \textbf{Empirical Evaluation.} Experiments show that RL with {\method}\texttt{-9B} achieves up to a 15.3\% relative EX improvement over GPT-5.4, while {\method}\texttt{-9B} also reaches 84.53\% top-1 failure-localization accuracy and 90.60\% MRR on held-out failures.
\end{enumerate}

\section{Preliminary}
In this section, we introduce the definition of Text-to-SQL and the objective of a process reward model.

\noindent{\textbf{Text-to-SQL.}} Text-to-SQL is a task that converts a natural language question $\mathcal{Q}$ into a SQL query $y$ capable of retrieving the correct answer from a database \cite{li2023bird}. Given a database $\mathcal{D} = \langle\mathcal{C},\mathcal{T}\rangle$, where $\mathcal{C}$ and $\mathcal{T}$ denote the sets of columns and tables, respectively, the Text-to-SQL task can be formulated as:
\begin{equation}
    y = f(\mathcal{Q}, \mathcal{D}\mid \theta),
\end{equation}
where $f(\cdot\mid\theta)$ is the text-to-SQL model parameterized by $\theta$. 

\noindent{\textbf{Process Reward Model.}} The process reward model (PRM) evaluates the correctness of each intermediate reasoning step, providing step-level supervision beyond the final answer \cite{lightman2023lets_verify}. Given an input $x$ and a step-by-step solution $z=(z_1,\dots,z_K)$, the reward of the $k$-th step can be formulated as:
\begin{equation}
    \label{eq:reward_func}
    f_{\psi}(x,z_{\le k})\in\mathbb{R},
    \quad
    q_{\psi}(\ell_k=\mathrm{positive}\mid x,z_{\le k})
    =
    \sigma\!\left(f_{\psi}(x,z_{\le k})\right),
\end{equation}
where $f_{\psi}$ is the real-valued PRM score parameterized by $\psi$, $q_{\psi}$ is the probability induced by the sigmoid link function, $z_{\le k}$ is the solution prefix up to step $k$, and $\ell_k$ is the process label of step $z_k$. A conventional PRM can rank a complete solution by multiplying step probabilities,
\begin{equation}
    R_{\mathrm{PRM}}^{\mathrm{rank}}(x,z)=
    \prod_{k=1}^{K} q_{\psi}(\ell_k=\mathrm{positive}\mid x,z_{\le k}),
\end{equation}
which estimates the probability that all intermediate steps are correct. In our RL formulation, we instead use the real-valued step scores as additive reward signals under the MDP return objective.

\section{Methodology}
\label{sec:method}
\begin{figure*}[h]
    \centering
    \includegraphics[width=\textwidth]{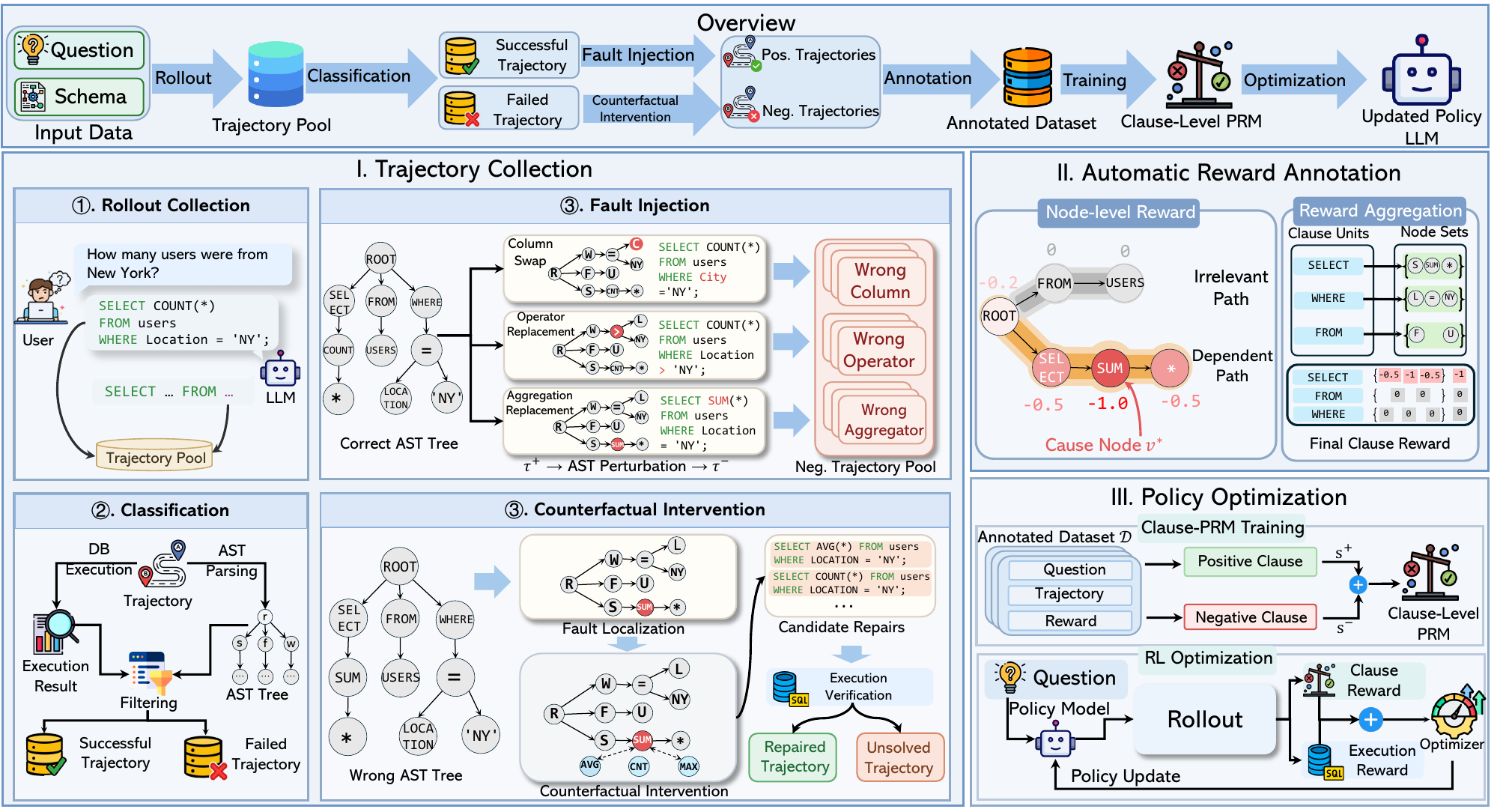}
    \caption{Overview of the {\method} framework.}
    \label{fig:framework}
\end{figure*}
In this section, we propose {\method}, which first constructs clause-level preference supervision from successful and failed SQL trajectories, then trains a Clause-Level Process Reward Model (Clause-PRM) on the resulting annotations, and finally uses it to provide clause-boundary rewards during Text-to-SQL policy optimization, as demonstrated in \Cref{fig:framework}.

\subsection{Clause-Level Text-to-SQL Formulation}
\label{sec:clause_formulation}
For policy optimization, we model SQL generation as an episodic MDP so that clause-aligned process scores can be incorporated into the trajectory return. At token step $t$, the state is the input and current SQL prefix $s_t=(x,a_{<t})$ with $x=(\mathcal{Q},\mathcal{D})$, and the action $a_t$ appends the next SQL token to the prefix. Under a policy $\pi_\theta(a_t\mid s_t)$, a rollout trajectory is
\begin{equation}
    \tau = (s_1,a_1,r_1,s_2,a_2,r_2,\dots,s_T,a_T,r_T,s_{T+1}).
\end{equation}
Let $G(\tau)=\sum_{t=1}^{T}\gamma^{t-1}r_t$ denote its return. The policy objective is
\begin{equation}
    \label{eq:text2sql_mdp}
    J(\theta)=\mathbb{E}_{\tau\sim\pi_\theta}[G(\tau)].
\end{equation}

Although rollouts are token-level, intermediate feedback should operate at the clause level, since SQL errors usually arise from a small number of incorrect top-level clause decisions. We therefore instantiate the generic PRM step in \Cref{eq:reward_func} as a top-level SQL clause unit.
\begin{definition}[Clause Unit]
    \label{def:clause}
A clause unit is a semantically coherent top-level SQL component.
We denote the $k$-th clause unit by $u_k^{\kappa_k}$, where $u_k$
is its content and $\kappa_k \in \mathcal{K}$ is its type:
\begin{equation}
\begin{aligned}
    \mathcal{K} = \{&
\textsc{With},
\textsc{Select},
\textsc{From-Join},
\textsc{Where},\\
&\textsc{Group-Having},
\textsc{Window},
\textsc{Order-Limit},
\textsc{Set-Op}
\}.
\end{aligned}
\end{equation}
Here, \textsc{With} covers CTEs; \textsc{From-Join} groups
\texttt{FROM}/\texttt{JOIN}/\texttt{ON}; \textsc{Group-Having}
groups \texttt{GROUP BY}/\texttt{HAVING}; \textsc{Window} covers
window functions and specifications; \textsc{Order-Limit} covers
\texttt{ORDER BY}/\texttt{LIMIT}/\texttt{OFFSET}; and \textsc{Set-Op} covers
\texttt{UNION}/\texttt{INTERSECT}/\texttt{EXCEPT}. We decompose only the top-level SQL structure; clauses inside nested subqueries are included in their enclosing clause unit.
\end{definition}

At the clause level, a generated SQL query is decomposed as
\begin{equation}
    y=(u_1^{\kappa_1},u_2^{\kappa_2},\dots,u_K^{\kappa_K}).
\end{equation}
Evaluating the $k$-th clause requires the preceding clause context, which we denote as
\begin{equation}
    p_k=(u_1^{\kappa_1},\dots,u_{k-1}^{\kappa_{k-1}}).
\end{equation}
We then instantiate the PRM in \Cref{eq:reward_func} as a real-valued clause preference scorer:
\begin{equation}
    \label{eq:clause_prm_func}
    f_{\psi}(x,p_k,u_k^{\kappa_k})\in\mathbb{R},
\end{equation}
where higher scores indicate clause units that are more consistent with the correct SQL continuation under $(x,p_k)$.

\subsection{Automatic Clause-Level Reward Annotation}
\label{sec:data_construction}
The annotation stage of {\method} converts coarse trajectory-level outcomes into fine-grained clause-level reward annotations.
We begin by collecting a trajectory pool from a supervised fine-tuned policy, denoted by $\pi_{\mathrm{SFT}}$. Let $\mathcal{X}=\{(x_i,y_i)\}_{i=1}^{N}$
be the set of training queries, where $x_i$ is the input question--schema pair and $y_i$ is the corresponding gold SQL query. For each input $x_i$, the SFT model generates one rollout trajectory $\tau_i \sim \pi_{\mathrm{SFT}}(\cdot \mid x_i)$,
which induces a terminal SQL query $\hat{y}(\tau_i)$. This trajectory collection is separate from tuple construction: a single retained rollout can yield multiple clause-level preference tuples because fault injection may perturb several editable nodes and counterfactual repair may proceed across multiple divergent clauses.
The collected trajectories are partitioned by their execution outcome. Let
\begin{equation}
    \Omega(\tau_i)=
\begin{cases}
1, & \text{if } \hat{y}(\tau_i) \text{ is execution-equivalent to } y_i,\\
0, & \text{otherwise}.
\end{cases}
\end{equation}
be the indicator function for trajectory classification. Accordingly, we define
\begin{equation}
    \mathfrak{T}^{+}=\{\tau_i:\Omega(\tau_i)=1\}, \text{ and } \mathfrak{T}^{-}=\{\tau_i:\Omega(\tau_i)=0\}.
\end{equation}
The successful set $\mathfrak{T}^{+}$ provides trusted positive trajectories for fault injection, whereas the failed set $\mathfrak{T}^{-}$ provides naturally occurring error trajectories for counterfactual intervention.
Our goal is to construct clause-level reward annotations by either inducing a root-cause error on $\mathfrak{T}^{+}$ or localizing a root-cause error on $\mathfrak{T}^{-}$.

For each trajectory $\tau$, we parse its terminal SQL $y(\tau)$ into an abstract syntax tree (AST) $G(\tau)=(V,E)$. According to \Cref{def:clause}, each generated unit is associated with a subset of AST nodes through a clause-to-node mapping $\phi(u_k^{\kappa_k}) = V_k \subseteq V$.
This mapping enables us to translate node-level structural errors on the AST into clause-level reward supervision for intermediate decisions along the trajectory.

\noindent{\textbf{Fault Injection for Successful Trajectories.}} For $\tau_i \in \mathfrak{T}^{+}$, the generated SQL is execution-correct and therefore provides a trusted positive reference. Using its AST, we sample one or more editable key nodes $v^\star \in V$ and apply a type-preserving mutation operator $\mu$ to obtain corrupted trees $G'=\mu(G,v^\star)$.
Each corrupted sample is retained only if its execution result differs from that of the original query, ensuring that the mutation induces a genuine semantic error.

\begin{wrapfigure}{r}{0.4\textwidth}
    \centering
    \includegraphics[width=0.4\textwidth]{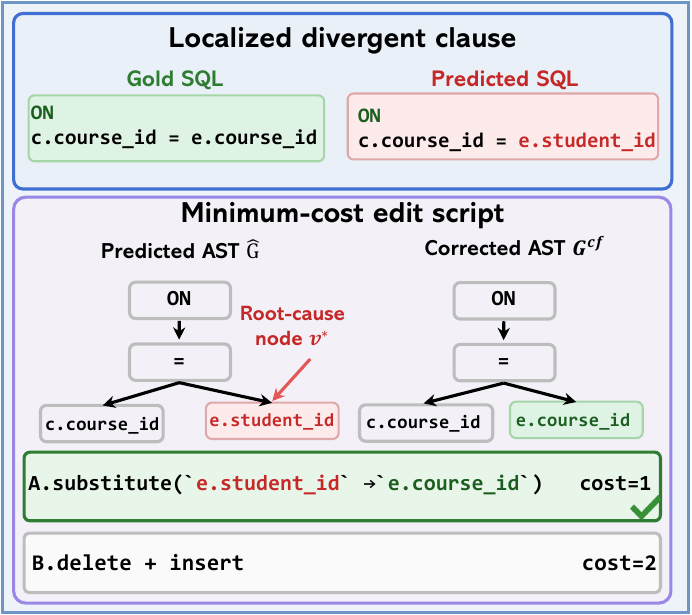}
    \caption{A counterfactual intervention example.}
    \label{fig:anchor_example}
\end{wrapfigure}
\noindent{\textbf{Counterfactual Intervention for Failed Trajectories.}} For $\tau_i \in \mathfrak{T}^{-}$, we recover clause-level supervision by comparing the predicted SQL against its gold counterpart. Let $u_k^{\kappa_k}$ and $\tilde u_k^{\kappa_k}$ denote the predicted and gold clause units. We first identify the earliest divergent clause index
\(
k^\star=\min\{k:\;u_k^{\kappa_k}\neq \tilde u_k^{\kappa_k}\}
\)
and then construct a counterfactual corrected query through an AST-level repair under the same prefix:
\begin{equation}
    y^{\mathrm{cf}}=\mathrm{Intervene}\!\left(\hat{y},\,p_{k^\star},\,u_{k^\star}^{\kappa_{k^\star}}\mapsto \tilde u_{k^\star}^{\kappa_{k^\star}}\right),
\end{equation}
where $\mathrm{Intervene}(\cdot)$ reparses the whole SQL after repair, canonicalizes aliases introduced by the repaired clause, and retains the counterfactual only when the resulting query is syntactically valid and executable. If a single repaired clause does not yield a valid counterfactual, we repeat on the earliest remaining divergence; each retained repair contributes a clause-level tuple. \Cref{fig:anchor_example} illustrates this process.

We then identify the node-level source of this mismatch. Let $\hat{G}=\mathrm{Parse}(\hat y)$ and $G^{\mathrm{cf}}=\mathrm{Parse}(y^{\mathrm{cf}})$ denote the ASTs of the predicted and counterfactually corrected queries, and let $\Xi_{k^\star}(\hat{G},G^{\mathrm{cf}})$ be the set of valid edit scripts that transform the divergent predicted clause into its corrected counterpart. We select the minimum-cost script and define the root-cause node $v^\star$ as the anchor of its first edit:
\begin{equation}
\xi^\star=\arg\min_{\xi \in \Xi_{k^\star}(\hat{G},G^{\mathrm{cf}})} \sum_{m=1}^{|\xi|} c(\xi_m),
\quad
v^\star=\mathrm{Anchor}(\xi_1^\star),
\label{eq:first_divergence}
\end{equation}
where $c(\xi_m)$ is the cost of the $m$-th node or edge edit, and $\mathrm{Anchor}(\xi_1^\star)$ returns the touched node in $\hat{G}$, or its parent attachment node for an insertion. Thus, $\xi^\star$ is the smallest local transformation from the predicted clause to its counterfactual correction.

\noindent{\textbf{Topology-Aware Node-Level Reward.}}
Once the root-cause node $v^\star$ is identified or induced, we assign a continuous penalty to each AST node $\hat v_i \in V$. \Cref{eq:first_divergence} localizes the clause-level mismatch anchor, while the distance below measures shortest-path distance to that anchor within $\hat{G}$:
\begin{equation}
R_{\mathrm{node}}(\hat v_i;v^\star)=
-\alpha \, M(\hat v_i,v^\star)
\exp\!\left(
-\frac{d_{\hat{G}}(\hat v_i,v^\star)^2}{2\sigma_d^2}
\right),
\label{eq:node_reward}
\end{equation}
where $M(\hat v_i,v^\star)\in\{0,1\}$ is a causal mask, and the graph distance $d_{\hat{G}}(\hat v_i,v^\star)=|\eta_{\hat{G}}(\hat v_i,v^\star)|$ is the shortest-path distance on $\hat{G}$, where $\eta_{\hat{G}}(\hat v_i,v^\star)$ is the unique simple AST path connecting $\hat v_i$ and $v^\star$, and $|\eta_{\hat{G}}|$ is the number of edges on that path. We set $M(\hat v_i,v^\star)=1$ only when $\hat v_i$ is $v^\star$ or an ancestor of $v^\star$ within the enclosing clause-level AST region, and $0$ otherwise. Here $\alpha>0$ is the maximum penalty magnitude, and $\sigma_d$ controls the distance-decay bandwidth. The root-cause node itself receives the strongest penalty $-\alpha$, while structurally upstream nodes on the same dependency path receive Gaussian-decayed negative reward. Unrelated branches are masked out and receive zero reward.

\noindent{\textbf{Clause-Level Aggregation and Supervision Construction.}}
To obtain supervision at the clause-unit level, we aggregate node-level penalties within each dispreferred unit:
\begin{equation}
R_{\mathrm{clause}}^{-}(u_k^{\kappa_k};v^\star)
=
\min_{\hat v_i \in \phi(u_k^{\kappa_k})}
R_{\mathrm{node}}(\hat v_i;v^\star).
\label{eq:clause_reward}
\end{equation}
The minimum operator propagates the strongest node-level penalty inside a clause unit to the clause level, ensuring that a clause containing the decisive error inherits a strong non-positive label.

For each supervised clause pair $(u_k^{+},u_k^{-})$, where $u_k^{+}$ is the preferred clause unit and $u_k^{-}$ is the dispreferred clause unit, we use the preferred clause as a zero-penalty reference and define the margin by the dispreferred clause penalty:
\begin{equation}
\delta_k = - R_{\mathrm{clause}}^{-}(u_k^{-};v^\star),
\label{eq:margin}
\end{equation}
Since $R_{\mathrm{clause}}^{-}(u_k^{-};v^\star)\le 0$ by construction, we have $\delta_k\ge 0$, with larger values indicating a more decisive structural error in the dispreferred clause.

Based on these clause-level rewards, we construct local supervision tuples under a shared input--prefix context. For successful trajectories, the original clause unit is preferred over its injected counterpart; for failed trajectories, the counterfactually corrected clause is preferred over the original failed clause. We therefore define
\begin{equation}
\mathcal{D}^{+}=
\left\{
(x,p_k,(u_k^{\kappa_k})^{+},(u_k^{\kappa_k})^{-},\kappa_k,\delta_k)
\;\middle|\;
\tau \in \mathfrak{T}^{+}
\right\},
\end{equation}
and
\begin{equation}
\mathcal{D}^{-}=
\left\{
(x,p_{k^\star},(\tilde u_{k^\star}^{\kappa_{k^\star}})^{+},(u_{k^\star}^{\kappa_{k^\star}})^{-},\kappa_{k^\star},\delta_{k^\star})
\;\middle|\;
\tau \in \mathfrak{T}^{-}
\right\}.
\end{equation}
The final annotated dataset is $\mathcal{D}=\mathcal{D}^{+}\cup\mathcal{D}^{-}$, where each tuple carries both a discrete preference label and a continuous margin $\delta_k$ reflecting the structural severity of the clause-level error.
\subsection{Clause-Level Process Reward Model}
\label{sec:clause_prm}
Using the clause-level PRM in \Cref{eq:clause_prm_func}, we train Clause-PRM on $\mathcal{D}$. For each tuple $(x,p_k,u_k^{+},u_k^{-},\kappa_k,\delta_k)\in\mathcal{D}$, we optimize the margin-weighted preference objective
\begin{equation}
    \mathcal{L}_{\mathrm{PRM}}
    =
    -\frac{1}{|\mathcal{D}|}
    \sum_{(x,p_k,u_k^{+},u_k^{-},\kappa_k,\delta_k)\in\mathcal{D}}
    \left(1+\frac{\delta_k}{\alpha}\right)
    \log \sigma\!\left(
    f_{\psi}(x,p_k,u_k^{+})-f_{\psi}(x,p_k,u_k^{-})
    \right),
    \label{eq:prm_objective}
\end{equation}
where $\sigma(\cdot)$ is the sigmoid function and $\alpha$ is the maximum penalty magnitude in \Cref{eq:node_reward}. Thus, clause pairs with larger structural margins contribute stronger preference supervision.

For policy optimization, the old policy samples a group of structured responses $\{z^{(g)}\}_{g=1}^{B}$ for the same input $x$. Each response yields a SQL query $\hat y^{(g)}=\mathrm{ExtractSQL}(z^{(g)})$ and a clause sequence $(\hat u_{g,1}^{\kappa_{g,1}},\dots,\hat u_{g,K_g}^{\kappa_{g,K_g}})$. At the $k$-th clause boundary, we assign
\begin{equation}
    r_{g,k}
    =
    I_{\mathrm{format}}(z^{(g)})
    \left[
    \lambda f_{\psi}(x,\hat p_{g,k},\hat u_{g,k}^{\kappa_{g,k}})
    +(1-\lambda)\mathbbm{1}[k=K_g]r_{\mathrm{exec}}^{(g)}
    \right],
    \label{eq:prefix_reward}
\end{equation}
where $r_{\mathrm{exec}}^{(g)}=\mathbbm{1}[\mathrm{Exec}(\hat y^{(g)})=\mathrm{Exec}(y)]$. Thus, the execution reward is added only once at the terminal clause boundary, while intermediate boundaries receive learned clause-level process rewards. The components are:
\begin{itemize}[leftmargin=*]
    \item \textbf{Format reward~\cite{li2025omnisql}:} $I_{\mathrm{format}}(z^{(g)})=\mathbbm{1}[\mathrm{ValidFormat}(z^{(g)})]$ checks the response structure rather than SQL syntax; it equals $1$ only when $z^{(g)}$ follows the required \texttt{<think>} $\cdots$ \texttt{</think>} and \texttt{<answer>} $\cdots$ \texttt{</answer>} format.
    \item \textbf{Clause reward:} $f_{\psi}(x,\hat p_{g,k},\hat u_{g,k}^{\kappa_{g,k}})$ scores the current clause under the predicted prefix.
    \item \textbf{Execution reward:} $r_{\mathrm{exec}}^{(g)}$ checks whether the predicted and gold SQL produce the same execution result.
\end{itemize}

Instead of collapsing all clause rewards into a single trajectory-level advantage, we compute a clause-level return-to-go $G_{g,k}= \sum_{j=k}^{K_g}\gamma_c^{j-k}r_{g,j}$ where $\gamma_c$ is a clause-level discount factor. We use GRPO-style normalization at each clause position. Let $\mathcal{B}_k=\{g:K_g\ge k\}$ be the valid responses that contain a $k$-th clause. The clause-level group advantage is
\begin{equation}
    A_{g,k}
    =
    \frac{
    G_{g,k}-\mathrm{mean}(\{G_{h,k}:h\in\mathcal{B}_k\})
    }{
    \mathrm{std}(\{G_{h,k}:h\in\mathcal{B}_k\})+\epsilon_{\mathrm{std}}
    },
    \label{eq:grpo_adv}
\end{equation}
for $g\in\mathcal{B}_k$. Let $c_g(t)\in\{1,\dots,K_g\}$ map token position $t$ to the clause span whose boundary reward supervises that token. The clipped GRPO objective becomes
\begin{equation}
    \mathcal{L}_{\mathrm{GRPO}}(\theta)
    =
    -\mathbb{E}_{x,\{z^{(g)}\}}
    \left[
    \frac{1}{B}
    \sum_{g=1}^{B}
    \frac{1}{|z^{(g)}|}
    \sum_{t=1}^{|z^{(g)}|}
    \min\!\left(
    \rho_{g,t}A_{g,c_g(t)},\,
    \bar\rho_{g,t}A_{g,c_g(t)}
    \right)
    \right],
    \label{eq:grpo_loss}
\end{equation}
where $\rho_{g,t}=\pi_{\theta}(z_t^{(g)}\mid x,z_{<t}^{(g)})/\pi_{\theta_{\mathrm{old}}}(z_t^{(g)}\mid x,z_{<t}^{(g)})$ and $\bar\rho_{g,t}=\mathrm{clip}(\rho_{g,t},1-\epsilon_{\mathrm{clip}},1+\epsilon_{\mathrm{clip}})$. This clause-level return-to-go preserves temporal credit assignment: tokens in a clause are optimized with the advantage of the remaining trajectory after that clause boundary, rather than a single trajectory-level advantage shared by all tokens.

\section{Experiments}
\subsection{Experiment Setup}
\noindent{\textbf{Data and Models.}} We evaluate {\method} on BIRD~\cite{li2023bird} and Spider~\cite{yu-etal-2018-spider}. We first run \texttt{Qwen3.5-9B} on the training splits and then apply counterfactual intervention and fault injection to build over $90{,}000$ clause-annotated preference tuples from BIRD-train, Spider-dev and SYNSQL-5k-train. We fine-tune \texttt{Qwen3.5-9B} on this data to obtain {\method}\texttt{-9B} as our clause-level PRM. For policy optimization, we start from a \texttt{Qwen3.5-9B-SFT} policy model (fine-tuned on BIRD-train and SYNSQL-5k \cite{ma2025sql}) and compare GRPO under sparse execution reward, token-level reward, and the clause-level reward from {\method}\texttt{-9B}.

\noindent{\textbf{Evaluation Protocol.}} For end-to-end Text-to-SQL, we report execution accuracy (EX) \cite{li2023bird} on the development set of BIRD, and both development and test sets of Spider under greedy decoding and Majority Vote@8 \cite{li2025omnisql}.
For failure localization, we hold out 10\% of the annotated data and evaluate only failed trajectories. Each method ranks the clause units in a trajectory, and we report top-1 localization accuracy ($\mathrm{Acc}_{\mathrm{loc}}$), Hit@3, and mean reciprocal rank (MRR). 

\noindent{\textbf{Baselines.}} For end-to-end evaluation, we compare against representative open- and closed-source Text-to-SQL LLMs, together with backbone-matched \texttt{Qwen3.5-9B} policies, including the supervised policy and GRPO variants under sparse, token, or clause-level rewards. For failure localization, we compare against heuristic selectors, prompted self-debugging, an execution-only reward model, and ablations of our annotation pipeline. For DeepEye-SQL, we use GPT-5.4 as the base model due to limited computational resources. Full baseline lists, hyperparameters, inference details, and hardware setup are provided in \Cref{app:exp_details}.

\subsection{Main Results}
Our main question is the impact of clause-level supervision on RL itself: \emph{does a more semantically aligned reward improve policy optimization over sparse execution feedback or token-level shaping?} Since Text-to-SQL is ultimately judged by end-to-end execution accuracy, we first evaluate Clause-PRM by its effect on the final GRPO-trained policy on BIRD and Spider.
As shown in \Cref{tab:ex_benchmarks}, three trends are clear. First, within the backbone-matched setting, \texttt{GRPO-Clause} is the strongest policy under both greedy decoding and Majority Vote@8, showing that clause-level feedback improves both single-sample quality and the quality of sampled candidate sets. Second, the gains are more pronounced under Majority Vote@8, suggesting that clause-level supervision helps the policy produce candidates with more reliable overall semantics, not just locally better next-token decisions. Third, simply making the reward denser is not enough: token-level reward does not match the clause-level variant, and in several cases it also fails to improve over sparse-reward GRPO. Overall, the trend across BIRD and Spider supports our central claim that semantically aligned clause-level credit assignment transfers to better end-to-end Text-to-SQL behavior.

\begin{table}[t]
\centering
\caption{Main execution accuracy (EX, \%) on BIRD and Spider under greedy decoding (P@1) and Majority Vote@8 (MV@8), where the backbone-matched block compares \texttt{Qwen3.5-9B-SFT} with GRPO variants trained under sparse, token-level, and clause-level rewards; the best and second-best results are highlighted in bold and underlined, respectively.}
\scriptsize
\begin{tabular}{lcccccccc}
\toprule
\multirow{2}{*}{Model} & \multicolumn{2}{c}{BIRD} & \multicolumn{2}{c}{Spider Dev} & \multicolumn{2}{c}{Spider Test} & \multicolumn{2}{c}{Avg.} \\
\cmidrule(lr){2-3}\cmidrule(lr){4-5}\cmidrule(lr){6-7}\cmidrule(lr){8-9}
& P@1 & MV@8 & P@1 & MV@8 & P@1 & MV@8 & P@1 & MV@8 \\
\midrule
\multicolumn{9}{c}{\textbf{Closed-source LLMs}} \\
\midrule
\texttt{Gemini-2.5-Pro}~\cite{comanici2025gemini} & \underline{62.32} & 63.77 & \textbf{82.78} & \textbf{84.13} & 78.04 & \underline{83.65} & \underline{74.38} & \underline{77.18} \\
\texttt{Claude Sonnet 4}~\cite{claude} & 55.34 & 66.10 & 76.40 & 78.52 & 75.04 & 79.97 & 68.93 & 74.86 \\
\texttt{GPT-5.4}~\cite{openai2026gpt54} & 59.58 & 60.36 & 77.85 & 83.26 & 78.24 & 81.73 & 71.89 & 75.12 \\
\texttt{Claude Sonnet 4.6}~\cite{anthropic2026sonnet46} & 54.62 & 67.13 & 76.59 & 79.01 & 74.84 & 78.59 & 68.68 & 74.91 \\
\texttt{DeepEye-SQL}~\cite{li2025deepeye_sql} & / & \underline{68.56} & / & 73.89 & / & 80.44 & / & 74.30 \\
\midrule
\multicolumn{9}{c}{\textbf{Open-source LLMs (7B--9B)}} \\
\midrule
\texttt{OmniSQL-7B}~\cite{li2025omnisql} & 36.83 & 56.38 & 76.78 & 79.11 & 77.55 & 79.69 & 63.72 & 71.73 \\
\texttt{SQL-R1-7B}~\cite{ma2025sql} & 36.57 & 56.58 & 78.52 & 79.21 & 80.34 & 80.48 & 65.14 & 72.09 \\
\texttt{XiyanSQL-7B}~\cite{gao2024xiyan_sql} & 30.64 & 35.00 & 55.03 & 79.49 & 53.88 & 76.99 & 46.52 & 63.83 \\
\texttt{AlphaSQL}~\cite{li2025alpha_sql} & / & 58.51 & / & 73.59  & / & 77.96& / & 70.02 \\
\texttt{Qwen3.5-9B}~\cite{qwen35} & 51.90 & 58.91 & 75.44 & 76.59 & 78.43 & 79.92 & 68.59 & 71.81 \\
\midrule
\multicolumn{9}{c}{\textbf{Open-source LLMs ($\geq$14B)}} \\
\midrule
\texttt{OmniSQL-14B}~\cite{li2025omnisql} & 45.96 & 65.71 & 76.98 & 79.20 & 78.34 & 79.83 & 67.09 & 74.91 \\
\texttt{XiyanSQL-14B}~\cite{gao2024xiyan_sql} & 37.15 & 42.11 & 58.46 & 80.27 & 73.87 & 78.44 & 56.49 & 66.94 \\
\midrule
\multicolumn{9}{c}{\textbf{Backbone-Matched Policies}} \\
\midrule
\texttt{Qwen3.5-9B-SFT}~\cite{qwen35} & 57.56 & 61.34 & 78.23 & 79.69 & 78.71 & 80.48 & 71.50 & 73.84 \\
\texttt{GRPO-Sparse} & 60.21 & 62.19 & 80.15 & 81.22 & \underline{81.56} & 82.09 & 73.97 & 75.17 \\
\texttt{GRPO-Clause} & \textbf{63.56} & \textbf{69.61} & \underline{81.63} & \underline{83.55} & \textbf{82.63} & \textbf{83.87} & \textbf{75.94} & \textbf{79.01} \\
\texttt{GRPO-Token} & 57.36 & 60.49 & 78.59 & 79.11 & 79.17 & 79.62 & 71.71 & 73.07 \\
\bottomrule
\end{tabular}
\label{tab:ex_benchmarks}
\end{table}

\subsection{Failure Localization}
We next evaluate whether Clause-PRM learns error attribution rather than only trajectory-level preference by asking each method to rank faulty clauses in held-out failed trajectories. As shown in \Cref{tab:failure_localization}, \texttt{{\method}-9B} achieves the best overall results across all three metrics. The strong \textsc{Self-Debug} baseline suggests that prompted LLMs can often identify suspicious regions, but Clause-PRM provides more reliable clause rankings through explicit clause-level supervision.
The ablations further clarify where this gain comes from. Removing topology causes a large drop in top-1 accuracy while keeping Hit@3 relatively high, which suggests that the model can still place the faulty clause near the top of the list but struggles to rank it precisely without structural propagation. Removing fault injection or counterfactual supervision is even more damaging, pushing performance much closer to heuristic and execution-only baselines. Together, these trends indicate that strong failure localization does not come from generic reward modeling alone; it depends on combining clause-level supervision with topology-aware propagation and counterfactual error construction.


\begin{table}[!ht]
\centering
\caption{Failure localization on held-out failures. All metrics are reported in percentages.}
\begin{tabular}{lccc}
\toprule
Method & $\mathrm{Acc}_{\mathrm{loc}}(\%) \uparrow$ & Hit@3 (\%) $\uparrow$ & MRR (\%) $\uparrow$ \\
\midrule
\textsc{Random-Clause} & 25.90 & 76.34 & 52.51 \\
\textsc{Last-Clause} & 9.82 & 53.49 & 38.60 \\
\texttt{Qwen3.5-9B} (\textsc{Self-Debug}) & \underline{82.78} & \underline{93.56} & \underline{86.35} \\
\texttt{Exec-RM-9B} & 25.93 & 71.50 & 50.82 \\
\texttt{w/o Topology} & 45.50 & 89.96 &67.34 \\
\texttt{w/o Fault Injection} & 22.20 & 80.17 & 51.24 \\
\texttt{w/o Counterfactual} & 26.02 & 83.06 & 54.31 \\
\midrule
\texttt{{\method}-9B} & \textbf{84.53} & \textbf{96.58} & \textbf{90.60} \\
\bottomrule
\end{tabular}
\label{tab:failure_localization}
\end{table}

\subsection{Clause-PRM as a Candidate Verifier}
\begin{figure}[h]
    \centering
    \includegraphics[width=0.85\textwidth]{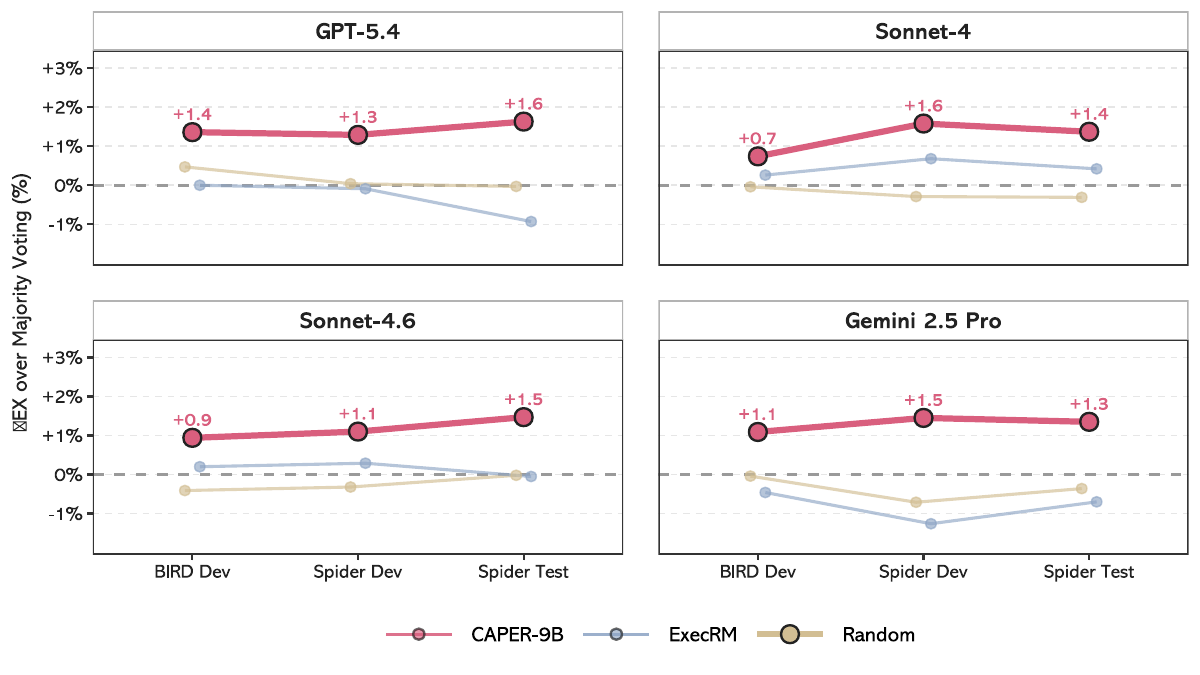}
    \caption{Candidate verification gains over Majority Vote@8, where each selector ranks the same eight sampled SQL candidates and $\Delta$EX measures the absolute execution-accuracy change over majority voting.}
    \label{fig:candidate_verifier_delta}
\end{figure}

 Beyond using Clause-PRM as a reward model for policy optimization, we further evaluate whether it can serve as a plug-and-play verifier for candidates generated by closed-source LLMs. This setting is motivated by black-box Text-to-SQL deployment, where practitioners often cannot update the generator but can sample multiple candidates and rerank them with an external verifier. For each question, all non-greedy selectors operate on the same eight sampled SQL candidates, so the comparison isolates the candidate selection rule rather than generation quality. As shown in \Cref{fig:candidate_verifier_delta}, \texttt{{\method}-9B} consistently improves over Majority Vote@8 across four closed-source generators and three evaluation splits, with gains ranging from $+0.7$ to $+1.6$ EX points. In contrast, \texttt{Exec-RM} and random selection are less stable and often stay near or below majority voting. This suggests that majority voting can be limited by execution-result frequency, while terminal-label reward models provide only coarse candidate-level signals; clause-level process supervision offers a more transferable semantic quality signal for identifying candidates with reliable intermediate SQL decisions.

\section{Related Work}

\noindent{\textbf{Text-to-SQL.}} Text-to-SQL has progressed from early semantic parsing systems such as Seq2SQL~\citep{zhong2017seq2sql} to cross-domain benchmarks like Spider~\citep{yu-etal-2018-spider}, with later work improving structural modeling through intermediate representations, relation-aware schema encoding, and constrained decoding~\citep{guo-etal-2019-irnet,wang-etal-2020-rat-sql,scholak-etal-2021-picard}. Recent studies further extend evaluation to more realistic settings, including BIRD and Spider~2.0, and increasingly rely on LLM-based decomposition and self-correction \citep{li2023bird,lei2024spider2,pourreza2023din-sql,gao2024dail-sql,he-etal-2025-star-sql}. Despite this progress, surveys consistently highlight persistent brittleness under schema shift and spurious language-SQL correlations \citep{deng-etal-2022-recent,liu2025text2sql-llm-era}. Our work is most related to this robustness direction, focusing on counterfactual supervision rather than new parser architectures or prompting pipelines.

\noindent{\textbf{Credit Assignment.}} Credit assignment remains difficult when rewards are sparse or delayed. Prior work addresses this problem through counterfactual attribution in logged feedback and RL, including Counterfactual Risk Minimization and its variants, COMA, and counterfactual off-policy evaluation~\citep{swaminathan2015crm,london2019bayesian-crm,zenati2023scrm,foerster2018coma,oberst2019counterfactual-ope}. Process supervision similarly shows that intermediate feedback can be more informative than pure outcome signals for long-horizon reasoning \citep{uesato2022process_feedback,lightman2023lets_verify}. In Text-to-SQL, recent work explores execution rewards, graph-based shaping, stepwise PRM guidance, and clause-wise critics~\citep{zhong2017seq2sql,ma2025sql,yao2025arctic_text2sql_r1,weng-etal-2025-graph,zhang2025reward_sql,chen2025sqlcritic}. Our method follows this credit-assignment perspective, but derives semantically grounded clause-level supervision from counterfactual AST perturbations for reward modeling and policy optimization.

\section{Conclusion}

This paper studies how to construct semantically meaningful process supervision for Text-to-SQL from coarse query-level outcomes. By attributing execution success and failure to clause-level SQL decisions, {\method} provides a practical credit assignment mechanism that avoids both sparse terminal labels and brittle token-level supervision. The resulting Clause-PRM improves failure localization, candidate verification, and GRPO-based policy optimization on BIRD and Spider, suggesting that clause-aligned process supervision is a scalable direction for structured generation.

\newpage
\bibliographystyle{unsrtnat}
\bibliography{reference}

@inproceedings{deng-etal-2022-recent,
  title = {Recent Advances in Text-to-SQL: A Survey of What We Have and What We Expect},
  author = {Deng, Naihao and Chen, Yulong and Zhang, Yue},
  booktitle = {Proceedings of the 29th International Conference on Computational Linguistics},
  pages = {2166--2187},
  address = {Gyeongju, Republic of Korea},
  publisher = {International Committee on Computational Linguistics},
  year = {2022},
  url = {https://aclanthology.org/2022.coling-1.190/}
}

@article{liu2025text2sql-llm-era,
  title = {A Survey of Text-to-SQL in the Era of LLMs: Where Are We, and Where Are We Going?},
  author = {Liu, Xinyu and Shen, Shuyu and Li, Boyan and Ma, Peixian and Jiang, Runzhi and Zhang, Yuxin and Fan, Ju and Li, Guoliang and Tang, Nan and Luo, Yuyu},
  journal = {IEEE Transactions on Knowledge and Data Engineering},
  volume = {37},
  number = {10},
  pages = {5735--5754},
  year = {2025},
  doi = {10.1109/TKDE.2025.3592032},
  url = {https://dblp.org/rec/journals/tkde/LiuSLMJZFLTL25}
}

@article{zhong2017seq2sql,
  title = {Seq2SQL: Generating Structured Queries from Natural Language Using Reinforcement Learning},
  author = {Zhong, Victor and Xiong, Caiming and Socher, Richard},
  journal = {CoRR},
  volume = {abs/1709.00103},
  year = {2017},
  doi = {10.48550/arXiv.1709.00103},
  url = {https://arxiv.org/abs/1709.00103}
}

@inproceedings{yu-etal-2018-spider,
  title = {Spider: A Large-Scale Human-Labeled Dataset for Complex and Cross-Domain Semantic Parsing and Text-to-SQL Task},
  author = {Yu, Tao and Zhang, Rui and Yang, Kai and Yasunaga, Michihiro and Wang, Dongxu and Li, Zifan and Ma, James and Li, Irene and Yao, Qingning and Roman, Shanelle and Zhang, Zilin and Radev, Dragomir},
  booktitle = {Proceedings of the 2018 Conference on Empirical Methods in Natural Language Processing},
  pages = {3911--3921},
  address = {Brussels, Belgium},
  publisher = {Association for Computational Linguistics},
  year = {2018},
  doi = {10.18653/v1/D18-1425},
  url = {https://aclanthology.org/D18-1425/}
}

@inproceedings{guo-etal-2019-irnet,
  title = {Towards Complex Text-to-SQL in Cross-Domain Database with Intermediate Representation},
  author = {Guo, Jiaqi and Zhan, Zecheng and Gao, Yan and Xiao, Yan and Lou, Jian-Guang and Liu, Ting and Zhang, Dongmei},
  booktitle = {Proceedings of the 57th Annual Meeting of the Association for Computational Linguistics},
  pages = {4524--4535},
  address = {Florence, Italy},
  publisher = {Association for Computational Linguistics},
  year = {2019},
  doi = {10.18653/v1/P19-1444},
  url = {https://aclanthology.org/P19-1444/}
}

@inproceedings{wang-etal-2020-rat-sql,
  title = {RAT-SQL: Relation-Aware Schema Encoding and Linking for Text-to-SQL Parsers},
  author = {Wang, Bailin and Shin, Richard and Liu, Xiaodong and Polozov, Oleksandr and Richardson, Matthew},
  booktitle = {Proceedings of the 58th Annual Meeting of the Association for Computational Linguistics},
  pages = {7567--7578},
  address = {Online},
  publisher = {Association for Computational Linguistics},
  year = {2020},
  doi = {10.18653/v1/2020.acl-main.677},
  url = {https://aclanthology.org/2020.acl-main.677/}
}

@inproceedings{scholak-etal-2021-picard,
  title = {PICARD: Parsing Incrementally for Constrained Auto-Regressive Decoding from Language Models},
  author = {Scholak, Torsten and Schucher, Nathan and Bahdanau, Dzmitry},
  booktitle = {Proceedings of the 2021 Conference on Empirical Methods in Natural Language Processing},
  pages = {9895--9901},
  address = {Online and Punta Cana, Dominican Republic},
  publisher = {Association for Computational Linguistics},
  year = {2021},
  doi = {10.18653/v1/2021.emnlp-main.779},
  url = {https://aclanthology.org/2021.emnlp-main.779/}
}

@article{gao2024dail-sql,
  title = {Text-to-SQL Empowered by Large Language Models: A Benchmark Evaluation},
  author = {Gao, Dawei and Wang, Haibin and Li, Yaliang and Sun, Xiuyu and Qian, Yichen and Ding, Bolin and Zhou, Jingren},
  journal = {Proceedings of the VLDB Endowment},
  volume = {17},
  number = {5},
  pages = {1132--1145},
  year = {2024},
  doi = {10.14778/3641204.3641221},
  url = {https://dblp.org/rec/journals/pvldb/GaoWLSQDZ24}
}

@inproceedings{pourreza2023din-sql,
	author = {Pourreza, Mohammadreza and Rafiei, Davood},
	booktitle = {Advances in Neural Information Processing Systems},
	pages = {36339--36348},
	title = {DIN-SQL: Decomposed In-Context Learning of Text-to-SQL with Self-Correction},
	url = {https://proceedings.neurips.cc/paper_files/paper/2023/file/72223cc66f63ca1aa59edaec1b3670e6-Paper-Conference.pdf},
	volume = {36},
	year = {2023},
}

@inproceedings{li2023bird,
  title = {Can LLM Already Serve as a Database Interface? A BIg Bench for Large-Scale Database Grounded Text-to-SQLs},
  author = {Li, Jinyang and Hui, Binyuan and Qu, Ge and Yang, Jiaxi and Li, Binhua and Li, Bowen and Wang, Bailin and Qin, Bowen and Geng, Ruiying and Huo, Nan and Zhou, Xuanhe and Ma, Chenhao and Li, Guoliang and Chang, Kevin Chen-Chuan and Huang, Fei and Cheng, Reynold and Li, Yongbin},
  booktitle = {Advances in Neural Information Processing Systems 36},
  year = {2023},
  url = {https://proceedings.neurips.cc/paper_files/paper/2023/hash/83fc8fab1710363050bbd1d4b8cc0021-Abstract-Datasets_and_Benchmarks.html}
}

@inproceedings{he-etal-2025-star-sql,
  title = {STaR-SQL: Self-Taught Reasoner for Text-to-SQL},
  author = {He, Mingqian and Shen, Yongliang and Zhang, Wenqi and Peng, Qiuying and Wang, Jun and Lu, Weiming},
  booktitle = {Proceedings of the 63rd Annual Meeting of the Association for Computational Linguistics (Volume 1: Long Papers)},
  pages = {24365--24375},
  address = {Vienna, Austria},
  publisher = {Association for Computational Linguistics},
  year = {2025},
  doi = {10.18653/v1/2025.acl-long.1187},
  url = {https://aclanthology.org/2025.acl-long.1187/}
}

@inproceedings{swaminathan2015crm,
  title = {Counterfactual Risk Minimization: Learning from Logged Bandit Feedback},
  author = {Swaminathan, Adith and Joachims, Thorsten},
  booktitle = {Proceedings of the 32nd International Conference on Machine Learning},
  series = {Proceedings of Machine Learning Research},
  volume = {37},
  pages = {814--823},
  address = {Lille, France},
  publisher = {PMLR},
  year = {2015},
  url = {https://proceedings.mlr.press/v37/swaminathan15.html}
}

@inproceedings{foerster2018coma,
  title = {Counterfactual Multi-Agent Policy Gradients},
  author = {Foerster, Jakob N. and Farquhar, Gregory and Afouras, Triantafyllos and Nardelli, Nantas and Whiteson, Shimon},
  booktitle = {Proceedings of the Thirty-Second AAAI Conference on Artificial Intelligence},
  pages = {2974--2982},
  publisher = {AAAI Press},
  year = {2018},
  doi = {10.1609/AAAI.V32I1.11794},
  url = {https://dblp.org/rec/conf/aaai/FoersterFANW18.html}
}

@inproceedings{london2019bayesian-crm,
  title = {Bayesian Counterfactual Risk Minimization},
  author = {London, Ben and Sandler, Ted},
  booktitle = {Proceedings of the 36th International Conference on Machine Learning},
  series = {Proceedings of Machine Learning Research},
  volume = {97},
  pages = {4125--4133},
  publisher = {PMLR},
  year = {2019},
  url = {https://proceedings.mlr.press/v97/london19a.html}
}

@inproceedings{oberst2019counterfactual-ope,
  title = {Counterfactual Off-Policy Evaluation with Gumbel-Max Structural Causal Models},
  author = {Oberst, Michael and Sontag, David},
  booktitle = {Proceedings of the 36th International Conference on Machine Learning},
  series = {Proceedings of Machine Learning Research},
  volume = {97},
  pages = {4881--4890},
  publisher = {PMLR},
  year = {2019},
  url = {https://proceedings.mlr.press/v97/oberst19a.html}
}

@inproceedings{zenati2023scrm,
  title = {Sequential Counterfactual Risk Minimization},
  author = {Zenati, Houssam and Diemert, Eustache and Martin, Matthieu and Mairal, Julien and Gaillard, Pierre},
  booktitle = {Proceedings of the 40th International Conference on Machine Learning},
  series = {Proceedings of Machine Learning Research},
  volume = {202},
  pages = {40681--40706},
  publisher = {PMLR},
  year = {2023},
  url = {https://proceedings.mlr.press/v202/zenati23a.html}
}

@article{uesato2022process_feedback,
  doi = {10.48550/ARXIV.2211.14275},
  url = {https://arxiv.org/abs/2211.14275},
  author = {Uesato, Jonathan and Kushman, Nate and Kumar, Ramana and Song, Francis and Siegel, Noah and Wang, Lisa and Creswell, Antonia and Irving, Geoffrey and Higgins, Irina},
  title = {Solving math word problems with process- and outcome-based feedback},
  publisher = {arXiv},
  year = {2022}
}

@inproceedings{lightman2023lets_verify,
  title={Let's verify step by step},
  author={Lightman, Hunter and Kosaraju, Vineet and Burda, Yuri and Edwards, Harrison and Baker, Bowen and Lee, Teddy and Leike, Jan and Schulman, John and Sutskever, Ilya and Cobbe, Karl},
  booktitle={The twelfth international conference on learning representations},
  year={2023}
}

@misc{yao2025arctic_text2sql_r1,
  doi = {10.48550/ARXIV.2505.20315},
  url = {https://arxiv.org/abs/2505.20315},
  author = {Yao, Zhewei and Sun, Guoheng and Borchmann, Lukasz and Nuti, Gaurav and Shen, Zheyu and Deng, Minghang and Zhai, Bohan and Zhang, Hao and Li, Ang and He, Yuxiong},
  title = {Arctic-Text2SQL-R1: Simple Rewards, Strong Reasoning in Text-to-SQL},
  publisher = {arXiv},
  year = {2025}
}

@misc{zhang2025reward_sql,
  doi = {10.48550/arXiv.2505.04671},
  url = {https://arxiv.org/abs/2505.04671},
  author = {Zhang, Yuxin and Fan, Meihao and Fan, Ju and Yi, Mingyang and Luo, Yuyu and Tan, Jian and Li, Guoliang},
  title = {Reward-{SQL}: Boosting Text-to-{SQL} via Stepwise Reasoning and Process-Supervised Rewards},
  publisher = {arXiv},
  year = {2025}
}

@misc{chen2025sqlcritic,
  doi = {10.48550/arXiv.2503.07996},
  url = {https://arxiv.org/abs/2503.07996},
  author = {Chen, Jikai and Gan, Leilei and Zhao, Ziyu and Wang, Zechuan and Wang, Dong and Zhuang, Chenyi},
  title = {{SQLCritic}: Correcting Text-to-{SQL} Generation via Clause-wise Critic},
  publisher = {arXiv},
  year = {2025}
}

@inproceedings{weng-etal-2025-graph,
  title = {Graph-Reward-{SQL}: Execution-Free Reinforcement Learning for Text-to-{SQL} via Graph Matching and Stepwise Reward},
  author = {Weng, Han and Wu, Puzhen and Longjie, Cui and Zhan, Yi and Liu, Boyi and Song, Yuanfeng and Zeng, Dun and Yang, Yingxiang and Zhang, Qianru and Huang, Dong and Yin, Xiaoming and Sun, Yang and Chen, Xing},
  booktitle = {Findings of the Association for Computational Linguistics: EMNLP 2025},
  month = nov,
  year = {2025},
  address = {Suzhou, China},
  publisher = {Association for Computational Linguistics},
  url = {https://aclanthology.org/2025.findings-emnlp.694/},
  doi = {10.18653/v1/2025.findings-emnlp.694},
  pages = {12917--12943},
  ISBN = {979-8-89176-335-7}
}

@article{hong2025next,
  title={Next-generation database interfaces: A survey of llm-based text-to-sql},
  author={Hong, Zijin and Yuan, Zheng and Zhang, Qinggang and Chen, Hao and Dong, Junnan and Huang, Feiran and Huang, Xiao},
  journal={IEEE Transactions on Knowledge and Data Engineering},
  year={2025},
  publisher={IEEE}
}

@inproceedings{lei2024spider2,
title={Spider 2.0: Evaluating Language Models on Real-World Enterprise Text-to-{SQL} Workflows},
author={Fangyu Lei and Jixuan Chen and Yuxiao Ye and Ruisheng Cao and Dongchan Shin and Hongjin SU and ZHAOQING SUO and Hongcheng Gao and Wenjing Hu and Pengcheng Yin and Victor Zhong and Caiming Xiong and Ruoxi Sun and Qian Liu and Sida Wang and Tao Yu},
booktitle={The Thirteenth International Conference on Learning Representations},
year={2025},
url={https://openreview.net/forum?id=XmProj9cPs}
}

@inproceedings{liswe,
	author = {Li, Jinyang and Li, Xiaolong and Qu, Ge and Jacobsson, Per and Qin, Bowen and Hui, Binyuan and Si, Shuzheng and Huo, Nan and Xu, Xiaohan and Zhang, Yue and Tang, Ziwei and Li, Yuanshuai and Widjaja, Florensia and Zhu, Xintong and Zhou, Feige and Huang, Yongfeng and Papakonstantinou, Yannis and Ozcan, Fatma and Chenhao, Ma and Cheng, Reynold},
	booktitle = {Advances in Neural Information Processing Systems},
	pages = {97085--97120},
	title = {SWE-SQL: Illuminating LLM Pathways to Solve User SQL Issues in Real-World Applications},
	url = {https://proceedings.neurips.cc/paper_files/paper/2025/file/8bfbf4ec87e1e331f0b1adc483b53b6b-Paper-Conference.pdf},
	volume = {38},
	year = {2025},
}

@inproceedings{huobird,
title={{BIRD}-{INTERACT}: Re-imagining Text-to-{SQL} Evaluation via Lens of Dynamic Interactions},
author={Nan Huo and Xiaohan Xu and Jinyang Li and Per Jacobsson and Shipei Lin and Bowen Qin and Binyuan Hui and Xiaolong Li and Ge Qu and Shuzheng Si and Linheng Han and Edward Alexander and Xintong Zhu and Rui Qin and Ruihan Yu and Yiyao Jin and Feige Zhou and Weihao Zhong and Yun Chen and Hongyu Liu and Chenhao Ma and Fatma Ozcan and Yannis Papakonstantinou and Reynold Cheng},
booktitle={The Fourteenth International Conference on Learning Representations},
year={2026},
url={https://openreview.net/forum?id=nHrYBGujps}
}

@article{luo2020deepeye,
  title={DeepEye: A Data Science System for Monitoring and Exploring COVID-19 Data.},
  author={Luo, Yuyu and Tang, Nan and Li, Guoliang and Li, Wenbo and Zhao, Tianyu and Yu, Xiang},
  journal={IEEE Data Eng. Bull.},
  volume={43},
  number={2},
  pages={121--132},
  year={2020}
}

@article{luo2020steerable,
  title={Steerable self-driving data visualization},
  author={Luo, Yuyu and Qin, Xuedi and Chai, Chengliang and Tang, Nan and Li, Guoliang and Li, Wenbo},
  journal={IEEE Transactions on Knowledge and Data Engineering},
  volume={34},
  number={1},
  pages={475--490},
  year={2020},
  publisher={IEEE}
}

@inproceedings{maddela2025starqa,
  title={STARQA: A Question Answering Dataset for Complex Analytical Reasoning over Structured Databases},
  author={Maddela, Mounica and Xie, Lingjue and Preo{\c{t}}iuc-Pietro, Daniel and others},
  booktitle={Proceedings of the 2025 Conference on Empirical Methods in Natural Language Processing},
  pages={34475--34487},
  year={2025}
}

@inproceedings{chen2025adept,
    title = "{ADEPT}-{SQL}: A High-performance Text-to-{SQL} Application for Real-World Enterprise-Level Databases",
    author = "Chen, Yongnan  and
      Chang, Zhuo  and
      Gu, Shijia  and
      Zong, Yuanhang  and
      Zhang, Mei  and
      Wang, Shiyu  and
      He, Zixiang  and
      Chen, HongZhi  and
      Jin, Wei  and
      Cui, Bin",
    editor = "Mishra, Pushkar  and
      Muresan, Smaranda  and
      Yu, Tao",
    booktitle = "Proceedings of the 63rd Annual Meeting of the Association for Computational Linguistics (Volume 3: System Demonstrations)",
    month = jul,
    year = "2025",
    address = "Vienna, Austria",
    publisher = "Association for Computational Linguistics",
    url = "https://aclanthology.org/2025.acl-demo.27/",
    doi = "10.18653/v1/2025.acl-demo.27",
    pages = "275--283",
    ISBN = "979-8-89176-253-4",
}

@inproceedings{xu2025abacus,
  title={Abacus-SQL: a text-to-SQL system empowering cross-domain and open-domain database retrieval},
  author={Xu, Keyan and Wang, Dingzirui and Zhang, Xuanliang and Zhu, Qingfu and Che, Wanxiang},
  booktitle={Proceedings of the 63rd Annual Meeting of the Association for Computational Linguistics (Volume 3: System Demonstrations)},
  pages={118--128},
  year={2025}
}

@inproceedings{zeng2020photon,
    title = "{P}hoton: A Robust Cross-Domain Text-to-{SQL} System",
    author = "Zeng, Jichuan  and
      Lin, Xi Victoria  and
      Hoi, Steven C.H.  and
      Socher, Richard  and
      Xiong, Caiming  and
      Lyu, Michael  and
      King, Irwin",
    editor = "Celikyilmaz, Asli  and
      Wen, Tsung-Hsien",
    booktitle = "Proceedings of the 58th Annual Meeting of the Association for Computational Linguistics: System Demonstrations",
    month = jul,
    year = "2020",
    address = "Online",
    publisher = "Association for Computational Linguistics",
    url = "https://aclanthology.org/2020.acl-demos.24/",
    doi = "10.18653/v1/2020.acl-demos.24",
    pages = "204--214",
}

@inproceedings{sheinin2018quest,
  title={Quest: a natural language interface to relational databases},
  author={Sheinin, Vadim and Khorashani, Elahe and Yeo, Hangu and Xu, Kun and Vo, Ngoc Phuoc An and Popescu, Octavian},
  booktitle={Proceedings of the Eleventh International Conference on Language Resources and Evaluation (LREC 2018)},
  year={2018}
}

@inproceedings{wang2025linkalign,
  title={Linkalign: Scalable schema linking for real-world large-scale multi-database text-to-sql},
  author={Wang, Yihan and Liu, Peiyu and Yang, Xin},
  booktitle={Proceedings of the 2025 Conference on Empirical Methods in Natural Language Processing},
  pages={977--991},
  year={2025}
}

@inproceedings{gan-etal-2023-appraising,
    title = "Re-appraising the Schema Linking for Text-to-{SQL}",
    author = "Gan, Yujian  and
      Chen, Xinyun  and
      Purver, Matthew",
    editor = "Rogers, Anna  and
      Boyd-Graber, Jordan  and
      Okazaki, Naoaki",
    booktitle = "Findings of the Association for Computational Linguistics: ACL 2023",
    month = jul,
    year = "2023",
    address = "Toronto, Canada",
    publisher = "Association for Computational Linguistics",
    url = "https://aclanthology.org/2023.findings-acl.53/",
    doi = "10.18653/v1/2023.findings-acl.53",
    pages = "835--852",
}

@inproceedings{kothyari2023crush4sql,
    title = "{CRUSH}4{SQL}: Collective Retrieval Using Schema Hallucination For {T}ext2{SQL}",
    author = "Kothyari, Mayank  and
      Dhingra, Dhruva  and
      Sarawagi, Sunita  and
      Chakrabarti, Soumen",
    editor = "Bouamor, Houda  and
      Pino, Juan  and
      Bali, Kalika",
    booktitle = "Proceedings of the 2023 Conference on Empirical Methods in Natural Language Processing",
    year = "2023",
    address = "Singapore",
    publisher = "Association for Computational Linguistics",
    url = "https://aclanthology.org/2023.emnlp-main.868/",
    doi = "10.18653/v1/2023.emnlp-main.868",
    pages = "14054--14066",
}

@inproceedings{wu2025ucs,
  title={UCS-SQL: uniting content and structure for enhanced semantic bridging in text-to-sql},
  author={Wu, Zhenhe and Li, Zhongqiu and Zhang, Jie and He, Zhongjiang and Yang, Jian and Zhao, Yu and Fang, Ruiyu and Wang, Bing and Xie, Hongyan and Song, Shuangyong and others},
  booktitle={Findings of the Association for Computational Linguistics: ACL 2025},
  pages={8156--8168},
  year={2025}
}

@inproceedings{xu2025ts,
  title={TS-SQL: Test-driven Self-refinement for Text-to-SQL},
  author={Xu, Wenbo and Zhu, Haifeng and Yan, Liang and Liu, Chuanyi and Han, Peiyi and Duan, Shaoming and Pan, Jeff Z},
  booktitle={Findings of the Association for Computational Linguistics: EMNLP 2025},
  pages={2864--2889},
  year={2025}
}

@inproceedings{qu2025share,
  title={SHARE: An SLM-based hierarchical action CorREction assistant for text-to-SQL},
  author={Qu, Ge and Li, Jinyang and Qin, Bowen and Li, Xiaolong and Huo, Nan and Ma, Chenhao and Cheng, Reynold},
  booktitle={Proceedings of the 63rd Annual Meeting of the Association for Computational Linguistics (Volume 1: Long Papers)},
  pages={11268--11292},
  year={2025}
}

@inproceedings{ma2025sql,
  title={SQL-R1: Training natural language to sql reasoning model by reinforcement learning},
  author={Peixian, Ma and Zhuang, Xialie and Xu, Chengjin and Jiang, Xuhui and Chen, Ran and Guo, Jian},
  booktitle={The Thirty-ninth Annual Conference on Neural Information Processing Systems},
  year={2025}
  }

@inproceedings{wolfson-etal-2022-weakly,
  title = "Weakly Supervised Text-to-{SQL} Parsing through Question Decomposition",
  author = "Wolfson, Tomer and Deutch, Daniel and Berant, Jonathan",
  booktitle = "Findings of the Association for Computational Linguistics: NAACL 2022",
  month = jul,
  year = "2022",
  address = "Seattle, United States",
  publisher = "Association for Computational Linguistics",
  url = "https://aclanthology.org/2022.findings-naacl.193/",
  doi = "10.18653/v1/2022.findings-naacl.193",
  pages = "2528--2542"
}

@inproceedings{zhang-etal-2025-exesql,
  title = "{E}xe{SQL}: Self-Taught Text-to-{SQL} Models with Execution-Driven Bootstrapping for {SQL} Dialects",
  author = "Zhang, Jipeng and Yang, Haolin and Miao, Kehao and Zhang, Ruiyuan and Pi, Renjie and Gao, Jiahui and Zhou, Xiaofang",
  booktitle = "Findings of the Association for Computational Linguistics: EMNLP 2025",
  month = nov,
  year = "2025",
  address = "Suzhou, China",
  publisher = "Association for Computational Linguistics",
  url = "https://aclanthology.org/2025.findings-emnlp.1320/",
  doi = "10.18653/v1/2025.findings-emnlp.1320",
  pages = "24305--24326",
  ISBN = "979-8-89176-335-7"
}

@article{li2025omnisql,
  title={OmniSQL: Synthesizing High-Quality Text-to-SQL Data at Scale},
  author={Li, Haoyang and Wu, Shang and Zhang, Xiaokang and Huang, Xinmei and Zhang, Jing and Jiang, Fuxin and Wang, Shuai and Zhang, Tieying and Chen, Jianjun and Shi, Rui and others},
  journal={Proceedings of the VLDB Endowment},
  volume={18},
  number={11},
  pages={4695--4709},
  year={2025},
  publisher={VLDB Endowment}
}

@misc{qwen35,
  author = {QwenTeam},
  title = {Qwen3.5: Towards Native Multimodal Agents},
  year = {2026},
  url = {https://qwen.ai/blog?id=qwen3.5},
}

@article{comanici2025gemini,
  title={Gemini 2.5: Pushing the frontier with advanced reasoning, multimodality, long context, and next generation agentic capabilities},
  author={Comanici, Gheorghe and Bieber, Eric and Schaekermann, Mike and Pasupat, Ice and Sachdeva, Noveen and Dhillon, Inderjit and Blistein, Marcel and Ram, Ori and Zhang, Dan and Rosen, Evan and others},
  journal={arXiv preprint arXiv:2507.06261},
  year={2025}
}

@misc{claude,
  author = {Anthropic},
  title = {{Claude Sonnet 4}},
  year = {2025},
  url = {https://www.anthropic.com/news/claude-4}
}

@misc{openai2026gpt54,
  author = {OpenAI},
  title = {Introducing {GPT}-5.4},
  year = {2026},
  month = mar,
  url = {https://openai.com/index/introducing-gpt-5-4/}
}

@misc{anthropic2026sonnet46,
  author = {Anthropic},
  title = {Introducing {Claude Sonnet 4.6}},
  year = {2026},
  month = feb,
  url = {https://www.anthropic.com/research/claude-sonnet-4-6}
}

@article{gao2024xiyan_sql,
  author={Liu, Yifu and Zhu, Yin and Gao, Yingqi and Luo, Zhiling and Li, Xiaoxia and Shi, Xiaorong and Hong, Yuntao and Gao, Jinyang and Li, Yu and Ding, Bolin and Zhou, Jingren},
  journal={IEEE Transactions on Knowledge and Data Engineering}, 
  title={XiYan-SQL: A Novel Multi-Generator Framework for Text-to-SQL}, 
  year={2026},
  volume={},
  number={},
  pages={1-14},
  doi={10.1109/TKDE.2026.3657851}
}

@inproceedings{li2025alpha_sql,
  title = {Alpha-{SQL}: Zero-Shot Text-to-{SQL} using {M}onte {C}arlo Tree Search},
  author = {Li, Boyan and Zhang, Jiayi and Fan, Ju and Xu, Yanwei and Chen, Chong and Tang, Nan and Luo, Yuyu},
  booktitle = {Proceedings of the 42nd International Conference on Machine Learning},
  pages = {36810--36830},
  year = {2025},
  volume = {267},
  series = {Proceedings of Machine Learning Research},
  publisher = {PMLR},
  url = {https://proceedings.mlr.press/v267/li25dt.html}
}

@article{li2025deepeye_sql,
  author = {Li, Boyan and Chen, Chong and Xue, Zhujun and Mei, Yinan and Luo, Yuyu},
  title = {DeepEye-SQL: A Software-Engineering-Inspired Text-to-SQL Framework},
  year = {2025},
  eprint = {2510.17586},
  archivePrefix = {arXiv},
  primaryClass = {cs.DB},
  doi = {10.48550/arXiv.2510.17586},
  url = {https://arxiv.org/abs/2510.17586}
}
\balance
\newpage
\appendix
\section{Limitations}
\label{sec:limitations}
While our proposed method demonstrates significant improvements in Text-to-SQL reinforcement learning, there are several limitations to consider. First, the clause-level reward model is trained on a specific set of datasets and may not generalize well to other domains or SQL dialects without further fine-tuning.  Second, we focus on SQL Generation rather than Chain of Thought (CoT) reasoning, and our method may not directly apply to other structured generation tasks that require different forms of intermediate supervision. Third, {\method} currently decomposes SQL only at the top-level clause granularity, which can still be coarse when the decisive error lies inside a clause, such as one predicate in a conjunctive \texttt{WHERE} condition or an internal subquery; although the same annotation procedure can in principle be applied recursively to obtain hierarchical rewards at the clause, sub-clause, and token levels, we restrict this work to the first level to keep the experimental cost manageable. Finally, the automated annotation pipeline introduces additional computational and engineering overhead because it requires SFT rollouts, AST parsing, counterfactual repair or fault injection, and repeated execution checks; this cost is higher than using sparse terminal execution rewards alone, although the resulting annotations can be reused for PRM training.

\section{Theoretical Analysis of Reward Propagation}
\label{app:reward_propagation_theory}

This section gives deterministic bounds for the topology-aware reward in \Cref{eq:node_reward} and the clause-level aggregation in \Cref{eq:clause_reward}. These results do not require any distributional assumption; they follow directly from the bounded mask, nonnegative AST distance, and the monotonicity of the Gaussian kernel.

\begin{lemma}[Bounds of Gaussian node propagation]
\label{lem:gaussian_node_bounds}
Let $\alpha>0$, $\sigma_d>0$, $M_i=M(\hat v_i,v^\star)\in\{0,1\}$, and $d_i=d_{\hat G}(\hat v_i,v^\star)\ge 0$. The node-level reward
\[
R_{\mathrm{node}}(\hat v_i;v^\star)
=
-\alpha M_i\exp\left(-\frac{d_i^2}{2\sigma_d^2}\right)
\]
satisfies
\[
-\alpha \le R_{\mathrm{node}}(\hat v_i;v^\star)\le 0.
\]
If $M_i=0$, then $R_{\mathrm{node}}(\hat v_i;v^\star)=0$. If $M_i=1$ and $d_i\le D$, then
\[
-\alpha
\le
R_{\mathrm{node}}(\hat v_i;v^\star)
\le
-\alpha\exp\left(-\frac{D^2}{2\sigma_d^2}\right).
\]
Moreover, for unmasked nodes, $R_{\mathrm{node}}$ is monotone nondecreasing in $d_i$, equivalently the error severity $-R_{\mathrm{node}}$ is monotone nonincreasing in $d_i$.
\end{lemma}

\begin{proof}
Since $d_i\ge 0$ and $\sigma_d>0$,
\[
0<\exp\left(-\frac{d_i^2}{2\sigma_d^2}\right)\le 1.
\]
Together with $M_i\in\{0,1\}$, this gives
\[
0
\le
M_i\exp\left(-\frac{d_i^2}{2\sigma_d^2}\right)
\le
1.
\]
Multiplying by $-\alpha$ yields $-\alpha \le R_{\mathrm{node}}(\hat v_i;v^\star)\le 0$. If $M_i=0$, the reward is exactly zero. If $M_i=1$ and $d_i\le D$, the Gaussian term is decreasing in $d_i$, so
\[
\exp\left(-\frac{D^2}{2\sigma_d^2}\right)
\le
\exp\left(-\frac{d_i^2}{2\sigma_d^2}\right)
\le
1.
\]
Multiplying by $-\alpha$ gives the stated local bound. Finally, when $M_i=1$,
\[
\frac{\partial R_{\mathrm{node}}}{\partial d_i}
=
\frac{\alpha d_i}{\sigma_d^2}
\exp\left(-\frac{d_i^2}{2\sigma_d^2}\right)
\ge 0,
\]
so the negative reward becomes weaker as the node moves farther from the anchor, while the corresponding severity $-R_{\mathrm{node}}$ decreases with distance.
\end{proof}

\begin{lemma}[Bounds of minimum clause aggregation]
\label{lem:min_clause_bounds}
For a nonempty clause unit $u$, define the active node set
\[
\mathcal{A}(u;v^\star)
=
\{\hat v_i\in\phi(u): M(\hat v_i,v^\star)=1\}.
\]
The clause-level reward in \Cref{eq:clause_reward} satisfies
\[
-\alpha
\le
R_{\mathrm{clause}}^{-}(u;v^\star)
\le
0.
\]
If $\mathcal{A}(u;v^\star)=\emptyset$, then $R_{\mathrm{clause}}^{-}(u;v^\star)=0$. Otherwise, let
\[
d_{\min}(u)
=
\min_{\hat v_i\in\mathcal{A}(u;v^\star)}
d_{\hat G}(\hat v_i,v^\star),
\qquad
D_u
=
\max_{\hat v_i\in\mathcal{A}(u;v^\star)}
d_{\hat G}(\hat v_i,v^\star).
\]
Then the minimum aggregation has the exact form
\[
R_{\mathrm{clause}}^{-}(u;v^\star)
=
-\alpha
\exp\left(-\frac{d_{\min}(u)^2}{2\sigma_d^2}\right),
\]
and therefore
\[
-\alpha
\le
R_{\mathrm{clause}}^{-}(u;v^\star)
\le
-\alpha
\exp\left(-\frac{D_u^2}{2\sigma_d^2}\right)
\le 0.
\]
\end{lemma}

\begin{proof}
By \Cref{lem:gaussian_node_bounds}, every node reward lies in $[-\alpha,0]$. The minimum over the nonempty set $\phi(u)$ must therefore also lie in $[-\alpha,0]$, proving the global bound.

If $\mathcal{A}(u;v^\star)=\emptyset$, then every node in the clause has mask value zero, hence every node reward is zero and the minimum is zero. If $\mathcal{A}(u;v^\star)\neq\emptyset$, inactive nodes have reward zero and cannot determine the minimum because active nodes have nonpositive rewards. Thus
\[
R_{\mathrm{clause}}^{-}(u;v^\star)
=
\min_{\hat v_i\in\mathcal{A}(u;v^\star)}
\left[
-\alpha
\exp\left(-\frac{d_{\hat G}(\hat v_i,v^\star)^2}{2\sigma_d^2}\right)
\right].
\]
The Gaussian term is decreasing in distance, and the leading negative sign means the most negative reward is achieved by the active node closest to $v^\star$. Therefore the minimum equals
\[
-\alpha
\exp\left(-\frac{d_{\min}(u)^2}{2\sigma_d^2}\right).
\]
Since $d_{\min}(u)\le D_u$, the same monotonicity gives the stated upper bound. Equivalently, if $s_i=-R_{\mathrm{node}}(\hat v_i;v^\star)\ge 0$ denotes node-level error severity, then
\[
-R_{\mathrm{clause}}^{-}(u;v^\star)
=
\max_{\hat v_i\in\phi(u)} s_i,
\]
so minimum reward aggregation is exactly maximum severity aggregation over nodes in the clause. This preserves the decisive local error rather than diluting it by clause length.
\end{proof}

\section{Algorithm Details}
The details of the annotation pipeline are shown in \Cref{alg:clause_annotation}.

\begin{algorithm}[!ht]
\caption{Automatic Clause-Level Reward Annotation}
\label{alg:clause_annotation}
\small
\renewcommand{\algorithmicrequire}{\textbf{Input:}}
\renewcommand{\algorithmicensure}{\textbf{Output:}}
\begin{algorithmic}[1]
\REQUIRE $\mathfrak{T}^{-},\ \mathfrak{T}^{+},\ \mu,\ \phi$
\ENSURE $\mathcal{D}$
\STATE $\mathcal{D}^{-}\leftarrow\emptyset,\quad \mathcal{D}^{+}\leftarrow\emptyset$
\STATE \textcolor{teal}{/* Part 1: Fault Injection on $\mathfrak{T}^{+}$ */}
\FOR{$\tau\in\mathfrak{T}^{+}$}
    \STATE $G\leftarrow \mathrm{Parse}(\hat y),\quad \mathcal{V}_{\mathrm{edit}}\leftarrow \mathrm{SampleEditable}(G)$
    \FOR{$v^\star\in\mathcal{V}_{\mathrm{edit}}$}
        \STATE $G'\leftarrow \mu(G,v^\star)$
        \IF{$\mathrm{Exec}(\mathrm{Decode}(G'))\neq \mathrm{Exec}(\mathrm{Decode}(G))$}
            \STATE $k\leftarrow \mathrm{ClauseIndex}(v^\star;\phi),\quad u_k^{+}\leftarrow [\mathrm{ClauseSplit}(\hat y)]_k,\quad u_k^{-}\leftarrow [\mathrm{ClauseSplit}(\mathrm{Decode}(G'))]_k$
            \STATE $\delta_k\leftarrow -R_{\mathrm{clause}}^{-}(u_k^{-};v^\star)$
            \STATE $\mathcal{D}^{+}\leftarrow \mathcal{D}^{+}\cup\{(x,p_k,u_k^{+},u_k^{-},\kappa_k,\delta_k)\}$
        \ENDIF
    \ENDFOR
\ENDFOR
\STATE \textcolor{teal}{/* Part 2: Counterfactual Intervention on $\mathfrak{T}^{-}$ */}
\FOR{$\tau\in\mathfrak{T}^{-}$}
    \STATE $x\leftarrow x(\tau),\ \hat y\leftarrow y(\tau),\ \tilde y\leftarrow y^{\mathrm{gt}}(\tau)$
    \STATE $\mathcal{K}_{\mathrm{div}}\leftarrow \{k:\ u_k\neq \tilde u_k\}$
    \FOR{$k^\star\in\mathcal{K}_{\mathrm{div}}$}
        \STATE $y^{\mathrm{cf}}\leftarrow \mathrm{Repair}(\hat y,\tilde y,p_{k^\star})$
        \IF{$\mathrm{Valid}(y^{\mathrm{cf}})$}
            \STATE $v^\star\leftarrow \mathrm{FirstDivergence}(\mathrm{Parse}(\hat y),\mathrm{Parse}(y^{\mathrm{cf}}))$
            \STATE $\tilde u_{k^\star}^{+}\leftarrow [\mathrm{ClauseSplit}(y^{\mathrm{cf}})]_{k^\star},\quad u_{k^\star}^{-}\leftarrow [\mathrm{ClauseSplit}(\hat y)]_{k^\star}$
            \STATE $\delta_{k^\star}\leftarrow -R_{\mathrm{clause}}^{-}(u_{k^\star}^{-};v^\star)$
            \STATE $\mathcal{D}^{-}\leftarrow \mathcal{D}^{-}\cup\{(x,p_{k^\star},\tilde u_{k^\star}^{+},u_{k^\star}^{-},\kappa_{k^\star},\delta_{k^\star})\}$
        \ENDIF
    \ENDFOR
\ENDFOR
\STATE \textbf{return} $\mathcal{D}\leftarrow \mathcal{D}^{-}\cup\mathcal{D}^{+}$
\end{algorithmic}
\end{algorithm}

\noindent{\textbf{Complexity Analysis.}}
Let $M=|\mathfrak{T}^{+}|+|\mathfrak{T}^{-}|$ be the number of trajectories to annotate, and let $N$ upper-bound the SQL length, number of clause units, AST size, and tokenized Clause-PRM input length for any trajectory. Clause splitting, AST traversal, fault injection, and clause-to-node reward aggregation are linear in $N$. The only super-linear algorithmic component is \textsc{FirstDivergence}, which performs clause-restricted edit-script computation; using a standard tree-edit dynamic program, this step costs $O(N^3)$ in the worst case. For database execution, let $D$ upper-bound the number of rows in any table and let $J$ be the maximum number of joined tables in a generated query. Without assuming a particular index or query plan, one execution has worst-case cost $O(D^J)$. Therefore, the worst-case offline annotation complexity is
\[
O\!\left(M(N^3 + D^J)\right).
\]
The storage complexity is $O(MN)$ when storing the retained clause preference tuples explicitly, or $O(M)$ if the tuples store references to the original SQL strings. During policy optimization, a rollout group of size $B$ contains at most $BN$ clause boundaries. A Transformer Clause-PRM forward pass over an $O(N)$-length input costs $O(N^2)$ when the model architecture is fixed, so reward scoring and terminal execution together add
\[
O\!\left(B(N^3 + D^J)\right),
\]
where the $BN^3$ term comes from $BN$ clause-boundary PRM calls and the $BD^J$ term comes from one terminal execution per rollout. In practice, clause-boundary PRM scores are batched, which improves wall-clock efficiency without changing the asymptotic bound.

\section{Experimental Setup Details}
\label{app:exp_details}

\noindent{\textbf{\method\texttt{-9B} Data Split.}} We train {\method}\texttt{-9B} on the final clause-annotated training split and use a held-out evaluation split for model selection and failure-localization evaluation. As shown in \Cref{tab:caper9b_data_stats}, the final split contains $82{,}081$ training rows and $9{,}166$ held-out evaluation rows, drawn from BIRD-train, SYNSQL-5k-train, and Spider-dev.

\begin{table}[ht]
\centering
\caption{{\method}\texttt{-9B} training and held-out evaluation data statistics.}
\label{tab:caper9b_data_stats}
\small
\begin{tabular}{lrrrr}
\toprule
Split & Total Rows & BIRD-Train & SYNSQL-5k-Train & Spider-Dev \\
\midrule
Train & $82{,}081$ & $58{,}090$ & $19{,}668$ & $4{,}323$ \\
Eval & $9{,}166$ & $6{,}488$ & $2{,}197$ & $481$ \\
\bottomrule
\end{tabular}
\end{table}

\noindent{\textbf{Training Configuration.}} We fine-tune \texttt{Qwen3.5-9B} on the clause-annotated dataset to obtain {\method}\texttt{-9B}. For policy optimization, we first train a supervised \texttt{Qwen3.5-9B-SFT} policy and then initialize GRPO from this checkpoint. We set the maximum penalty magnitude $\alpha$ in \Cref{eq:node_reward} to 1.0, the Gaussian distance bandwidth $\sigma_d$ in \Cref{eq:node_reward} to 1.0, and the balance factor $\lambda$ in \Cref{eq:prefix_reward} to 0.5. RL training uses batch size 16, 8 rollouts, and learning rate $1 \times 10^{-6}$.

\noindent{\textbf{\texttt{GRPO-Token} Baseline.}} \texttt{GRPO-Token} uses the same initialization checkpoint, GRPO objective, rollout number, batch size, learning rate, and terminal execution reward as \texttt{GRPO-Clause}, but replaces the learned clause-level process reward with heuristic token-level dense shaping. Specifically, after extracting the SQL from the \texttt{<answer>} field, we canonicalize keyword casing and whitespace, tokenize the generated and gold SQL strings, and assign a token process reward $\mathbbm{1}[\hat{s}_t=s_t^\star]$ at each SQL-token position $t$ when both positions exist, and $0$ otherwise. The terminal execution reward is still added only once at the final token. This baseline isolates the effect of making rewards dense at the token level without training an additional token-level PRM.

\noindent{\textbf{Inference and Metrics.}} For end-to-end evaluation, we report execution accuracy (EX) \cite{li2023bird}. Given $N$ evaluation examples with predicted SQL $\hat y_i$ and gold SQL $y_i$,
\begin{equation}
\mathrm{EX}=\frac{1}{N}\sum_{i=1}^{N}\mathbbm{1}\!\left[\mathrm{Exec}(\hat y_i)=\mathrm{Exec}(y_i)\right].
\end{equation}
Greedy decoding with temperature 0 corresponds to P@1. For Majority Vote@8, we sample candidates $\{\hat y_i^{(m)}\}_{m=1}^{8}$ with temperature 0.8, group them by execution result, and choose a candidate from the largest group:
\begin{equation}
r_i^{\mathrm{MV}}=\arg\max_{r}\sum_{m=1}^{8}\mathbbm{1}\!\left[\mathrm{Exec}(\hat y_i^{(m)})=r\right],
\qquad
\hat y_i^{\mathrm{MV}@8}\in\{\hat y_i^{(m)}:\mathrm{Exec}(\hat y_i^{(m)})=r_i^{\mathrm{MV}}\}.
\end{equation}
The reported MV@8 is $\mathrm{EX}$ computed with $\hat y_i=\hat y_i^{\mathrm{MV}@8}$. For candidate verification, we report the absolute gain over majority voting,
\begin{equation}
\Delta\mathrm{EX}=\mathrm{EX}(\hat y^{\mathrm{selector}})-\mathrm{EX}(\hat y^{\mathrm{MV}@8}).
\end{equation}
For failure localization, each method ranks the clause units in a failed trajectory. Let $\pi_i$ denote the predicted ranking from most to least likely faulty for held-out failure $i$, $k_i^\star$ denote the annotated faulty clause, and $\rho_i=\mathrm{rank}_{\pi_i}(k_i^\star)$. For $M$ held-out failures, we compute
\begin{equation}
\mathrm{Acc}_{\mathrm{loc}}=\frac{1}{M}\sum_{i=1}^{M}\mathbbm{1}[\rho_i=1],
\qquad
\mathrm{Hit@3}=\frac{1}{M}\sum_{i=1}^{M}\mathbbm{1}[\rho_i\le 3],
\qquad
\mathrm{MRR}=\frac{1}{M}\sum_{i=1}^{M}\frac{1}{\rho_i}.
\end{equation}

\noindent{\textbf{Full Baseline List.}} For end-to-end evaluation, we compare against representative Text-to-SQL models of varying scales, including \texttt{Gemini-2.5-Pro} \cite{comanici2025gemini}, \texttt{Claude Sonnet 4} \cite{claude}, \texttt{GPT-5.4} \cite{openai2026gpt54}, \texttt{Claude Sonnet 4.6} \cite{anthropic2026sonnet46}, \texttt{DeepEye-SQL} \cite{li2025deepeye_sql}, \texttt{OmniSQL-7B} \cite{li2025omnisql}, \texttt{SQL-R1-7B} \cite{ma2025sql}, \texttt{XiyanSQL-7B}~\cite{gao2024xiyan_sql}, \texttt{AlphaSQL} \cite{li2025alpha_sql}, \texttt{Qwen3.5-9B} \cite{qwen35}, \texttt{OmniSQL-14B}~\cite{li2025omnisql}, and \texttt{XiyanSQL-14B} \cite{gao2024xiyan_sql}. Within the backbone-matched policy block, we additionally compare the supervised \texttt{Qwen3.5-9B-SFT} policy and GRPO variants trained under sparse, clause-level, and token-level rewards, as listed in Table~\ref{tab:ex_benchmarks}. For failure localization, we compare \texttt{{\method}-9B} against four classes of baselines: (1) heuristic selectors, including \textsc{Random-Clause} and \textsc{Last-Clause}; (2) prompted self-debugging with \texttt{Qwen3.5-9B}, which directly predicts the faulty clause from the failed trajectory; (3) an execution-only reward model \texttt{Exec-RM-9B} trained from terminal execution labels without clause-level supervision; and (4) ablated variants of our approach, including \texttt{w/o Topology}, \texttt{w/o Fault Injection}, and \texttt{w/o Counterfactual Intervention}.

\noindent{\textbf{Environment.}} All experiments are conducted on a server with 8 NVIDIA A100 (80GB) GPUs. RL training is implemented with OpenRLHF\footnote{https://github.com/openrlhf/openrlhf}. We summarize the measured compute for the main local training and evaluation stages in \Cref{tab:compute_resources}; GPU-hours are computed as the number of GPUs multiplied by wall-clock hours.

\begin{table}[ht]
\centering
\caption{Compute resources for the main local training and evaluation stages.}
\label{tab:compute_resources}
\small
\begin{tabular}{lccc}
\toprule
Stage & Hardware & Wall-clock Time & GPU-hours \\
\midrule
Clause-PRM training & 8 A100 (80GB) & 4 h 39 m & 37.2 \\
GRPO policy optimization & 8 A100 (80GB) & 19 h 1 m 46 s & 152.2 \\
Evaluation & 2 A100 (80GB) & 2 h & 4.0 \\
\bottomrule
\end{tabular}
\end{table}

\noindent{\textbf{Existing Assets and Licenses.}} We use public datasets, model checkpoints, software, and API services under their stated licenses or access terms, summarized in \Cref{tab:existing_assets}. For derived artifacts released with this submission, the supplemental code package includes a README with asset links, preprocessing steps, release notes, and applicable terms.

\begin{table}[ht]
\centering
\caption{Existing assets used in this work and their licenses or access terms.}
\label{tab:existing_assets}
\scriptsize
\begin{tabular}{p{0.25\textwidth}p{0.30\textwidth}p{0.36\textwidth}}
\toprule
Asset & Use in this work & License / access terms \\
\midrule
BIRD~\cite{li2023bird} & Training and evaluation data & \href{https://bird-bench.github.io/}{CC BY-SA 4.0} \\
Spider~\cite{yu-etal-2018-spider} & Training and evaluation data & \href{https://github.com/taoyds/spider}{Apache-2.0} \\
SynSQL-Complex-5K~\cite{ma2025sql} & Synthetic training data & \href{https://huggingface.co/datasets/MPX0222forHF/SynSQL-Complex-5K}{Apache-2.0} \\
\texttt{Qwen3.5-9B}~\cite{qwen35} & Base policy and PRM backbone & \href{https://huggingface.co/Qwen/Qwen3.5-9B}{Apache-2.0} \\
\texttt{OmniSQL-7B/14B}~\cite{li2025omnisql} & Open-source Text-to-SQL baselines & Apache-2.0 Hugging Face releases: \href{https://huggingface.co/seeklhy/OmniSQL-7B}{7B}, \href{https://huggingface.co/seeklhy/OmniSQL-14B}{14B} \\
\texttt{SQL-R1-7B}~\cite{ma2025sql} & Open-source Text-to-SQL baseline & \href{https://huggingface.co/MPX0222forHF/SQL-R1-7B}{Apache-2.0} \\
\texttt{XiYanSQL-7B/14B}~\cite{gao2024xiyan_sql} & Open-source Text-to-SQL baselines & Apache-2.0 Hugging Face releases: \href{https://huggingface.co/XGenerationLab/XiYanSQL-QwenCoder-7B-2504}{7B}, \href{https://huggingface.co/XGenerationLab/XiYanSQL-QwenCoder-14B-2504}{14B} \\
\texttt{AlphaSQL}~\cite{li2025alpha_sql} & Baseline system and reported comparison & \href{https://github.com/HKUSTDial/Alpha-SQL}{MIT} \\
\texttt{DeepEye-SQL}~\cite{li2025deepeye_sql} & Baseline system and reported comparison & \href{https://github.com/HKUSTDial/DeepEye-SQL}{MIT} \\
OpenRLHF & RL training framework & \href{https://github.com/OpenRLHF/OpenRLHF}{Apache-2.0} \\
GPT, Claude, and Gemini APIs~\cite{openai2026gpt54,claude,anthropic2026sonnet46,comanici2025gemini} & Closed-source LLM baselines and candidate generation & Provider terms: \href{https://openai.com/policies/service-terms/}{OpenAI Service Terms}, \href{https://www.anthropic.com/legal/commercial-terms}{Anthropic Commercial Terms}, \href{https://ai.google.dev/gemini-api/terms}{Gemini API Additional Terms} \\
\bottomrule
\end{tabular}
\end{table}

\section{Candidate Verification Case Study}
\label{app:candidate_verification_case}

We provide a qualitative example from BIRD Dev to illustrate why clause-level verification can be more reliable than selecting candidates by sampling frequency or terminal-only reward scores. As shown in \Cref{fig:candidate_verification_case}, the question asks for the Italian flavor text of the card \texttt{Ancestor's Chosen} in the \texttt{card\_games} database. The gold query must join \texttt{cards} with \texttt{foreign\_data} through \texttt{uuid}, because the card name is stored in \texttt{cards}, whereas the localized flavor text and language are stored in \texttt{foreign\_data}.

\begin{figure}[!ht]
    \centering
    \includegraphics[width=0.98\textwidth]{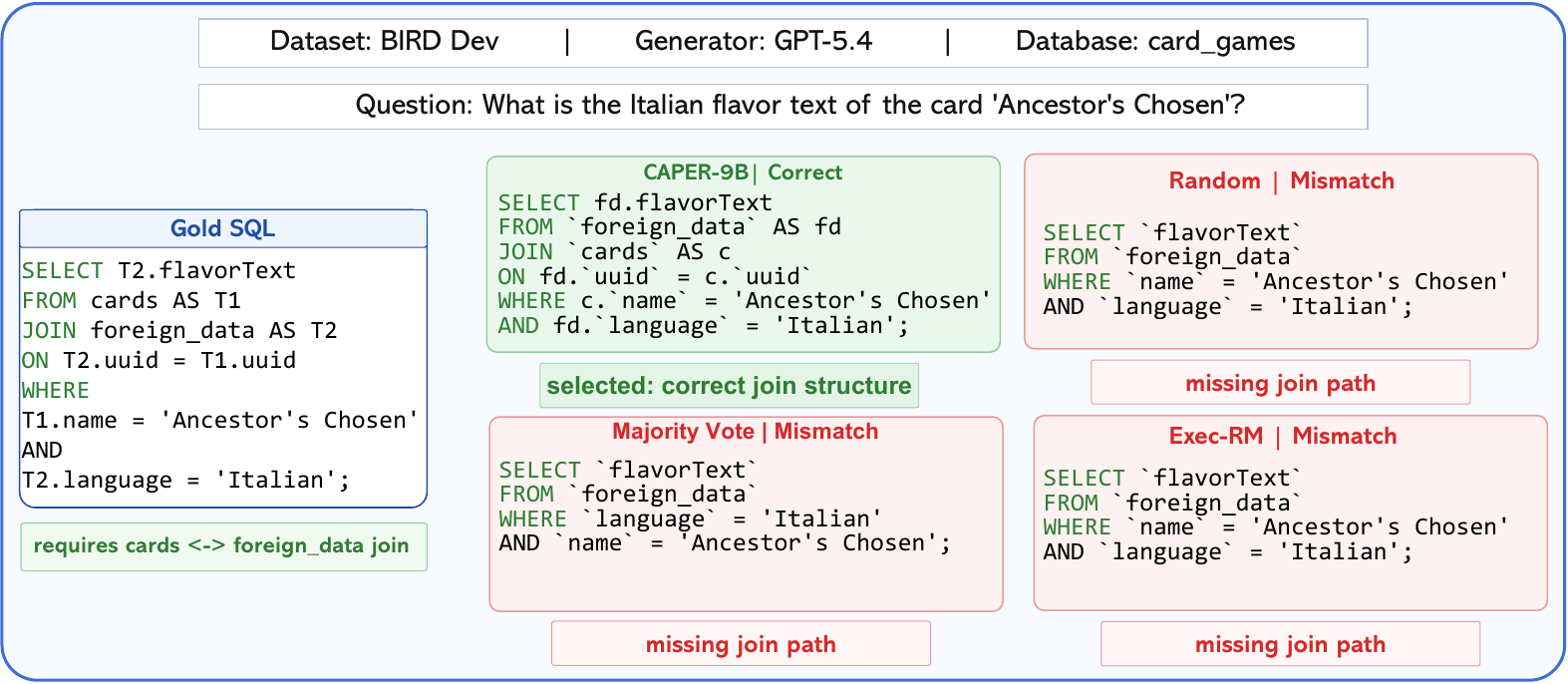}
    \caption{Candidate verification case study on BIRD Dev with GPT-5.4 candidates. Random, Majority Vote@8, and \texttt{Exec-RM} select candidates that directly filter \texttt{foreign\_data} by card name and therefore miss the required \texttt{cards}--\texttt{foreign\_data} join. \texttt{{\method}-9B} selects the only candidate that matches the gold join structure.}
    \label{fig:candidate_verification_case}
\end{figure}

This example exposes a common failure mode for simpler selectors. The incorrect candidates mention the requested output column and language condition, but they omit the schema-level dependency between the English card name and its localized foreign record. As a result, Random, Majority Vote@8, and \texttt{Exec-RM} all choose candidates that fail execution correctness. In contrast, \texttt{{\method}-9B} selects the candidate that preserves the required clause structure: the join between \texttt{cards} and \texttt{foreign\_data}, the name filter on \texttt{cards}, and the language filter on \texttt{foreign\_data}. This case suggests that clause-level process supervision provides a transferable structural quality signal when terminal-only or frequency-based selectors prefer superficially plausible SQL.

\section{Evaluation Prompt Example}
\label{app:prompt_example}

\begin{tcblisting}{
title=Example Evaluation Prompt,
enhanced,
breakable,
colback=white,
colframe=blue!55!black,
colbacktitle=blue!55!black,
coltitle=white,
fonttitle=\bfseries\large,
boxrule=0.9pt,
arc=0pt,
outer arc=0pt,
left=10pt,
right=10pt,
top=8pt,
bottom=8pt,
toptitle=4pt,
bottomtitle=4pt,
drop shadow,
listing only,
listing options={basicstyle=\ttfamily\scriptsize,breaklines=true,breakatwhitespace=false,columns=fullflexible}
}
Task Overview:
You are a data science expert. Below, you are provided with a database
schema and a natural language question. Your task is to understand the
schema and generate a valid SQL query to answer the question.

Database Engine:
SQLite

Database ID:
formula_1

Database Schema:
CREATE TABLE `circuits` (
  `circuitId` INTEGER,
  `circuitRef` TEXT,
  `name` TEXT,
  `location` TEXT,
  `country` TEXT,
  `lat` REAL,
  `lng` REAL,
  `alt` INTEGER,
  `url` TEXT,
  PRIMARY KEY (`circuitId`)
);
CREATE TABLE `constructors` (
  `constructorId` INTEGER,
  `constructorRef` TEXT,
  `name` TEXT,
  `nationality` TEXT,
  `url` TEXT,
  PRIMARY KEY (`constructorId`)
);
CREATE TABLE `drivers` (
  `driverId` INTEGER,
  `driverRef` TEXT,
  `number` INTEGER,
  `code` TEXT,
  `forename` TEXT,
  `surname` TEXT,
  `dob` DATE,
  `nationality` TEXT,
  `url` TEXT,
  PRIMARY KEY (`driverId`)
);
CREATE TABLE `seasons` (
  `year` INTEGER,
  `url` TEXT,
  PRIMARY KEY (`year`)
);
CREATE TABLE `races` (
  `raceId` INTEGER,
  `year` INTEGER,
  `round` INTEGER,
  `circuitId` INTEGER,
  `name` TEXT,
  `date` DATE,
  `time` TEXT,
  `url` TEXT,
  PRIMARY KEY (`raceId`)
);
CREATE TABLE `constructorResults` (
  `constructorResultsId` INTEGER,
  `raceId` INTEGER,
  `constructorId` INTEGER,
  `points` REAL,
  `status` TEXT,
  PRIMARY KEY (`constructorResultsId`)
);
CREATE TABLE `constructorStandings` (
  `constructorStandingsId` INTEGER,
  `raceId` INTEGER,
  `constructorId` INTEGER,
  `points` REAL,
  `position` INTEGER,
  `positionText` TEXT,
  `wins` INTEGER,
  PRIMARY KEY (`constructorStandingsId`)
);
CREATE TABLE `driverStandings` (
  `driverStandingsId` INTEGER,
  `raceId` INTEGER,
  `driverId` INTEGER,
  `points` REAL,
  `position` INTEGER,
  `positionText` TEXT,
  `wins` INTEGER,
  PRIMARY KEY (`driverStandingsId`)
);
CREATE TABLE `lapTimes` (
  `raceId` INTEGER,
  `driverId` INTEGER,
  `lap` INTEGER,
  `position` INTEGER,
  `time` TEXT,
  `milliseconds` INTEGER,
  PRIMARY KEY (`raceId`, `driverId`, `lap`)
);
CREATE TABLE `pitStops` (
  `raceId` INTEGER,
  `driverId` INTEGER,
  `stop` INTEGER,
  `lap` INTEGER,
  `time` TEXT,
  `duration` TEXT,
  `milliseconds` INTEGER,
  PRIMARY KEY (`raceId`, `driverId`, `stop`)
);
CREATE TABLE `qualifying` (
  `qualifyId` INTEGER,
  `raceId` INTEGER,
  `driverId` INTEGER,
  `constructorId` INTEGER,
  `number` INTEGER,
  `position` INTEGER,
  `q1` TEXT,
  `q2` TEXT,
  `q3` TEXT,
  PRIMARY KEY (`qualifyId`)
);
CREATE TABLE `status` (
  `statusId` INTEGER,
  `status` TEXT,
  PRIMARY KEY (`statusId`)
);
CREATE TABLE `results` (
  `resultId` INTEGER,
  `raceId` INTEGER,
  `driverId` INTEGER,
  `constructorId` INTEGER,
  `number` INTEGER,
  `grid` INTEGER,
  `position` INTEGER,
  `positionText` TEXT,
  `positionOrder` INTEGER,
  `points` REAL,
  `laps` INTEGER,
  `time` TEXT,
  `milliseconds` INTEGER,
  `fastestLap` INTEGER,
  `rank` INTEGER,
  `fastestLapTime` TEXT,
  `fastestLapSpeed` TEXT,
  `statusId` INTEGER,
  PRIMARY KEY (`resultId`)
);
-- Foreign Keys:
-- `races`.`circuitId` can be joined with `circuits`.`circuitId`
-- `races`.`year` can be joined with `seasons`.`year`
-- `constructorResults`.`constructorId` can be joined with `constructors`.`constructorId`
-- `constructorResults`.`raceId` can be joined with `races`.`raceId`
-- `constructorStandings`.`constructorId` can be joined with `constructors`.`constructorId`
-- `constructorStandings`.`raceId` can be joined with `races`.`raceId`
-- `driverStandings`.`driverId` can be joined with `drivers`.`driverId`
-- `driverStandings`.`raceId` can be joined with `races`.`raceId`
-- `lapTimes`.`driverId` can be joined with `drivers`.`driverId`
-- `lapTimes`.`raceId` can be joined with `races`.`raceId`
-- `pitStops`.`driverId` can be joined with `drivers`.`driverId`
-- `pitStops`.`raceId` can be joined with `races`.`raceId`
-- `qualifying`.`constructorId` can be joined with `constructors`.`constructorId`
-- `qualifying`.`driverId` can be joined with `drivers`.`driverId`
-- `qualifying`.`raceId` can be joined with `races`.`raceId`
-- `results`.`statusId` can be joined with `status`.`statusId`
-- `results`.`constructorId` can be joined with `constructors`.`constructorId`
-- `results`.`driverId` can be joined with `drivers`.`driverId`
-- `results`.`raceId` can be joined with `races`.`raceId`
This schema describes the database's structure, including tables, columns,
primary keys, foreign keys, and any relevant relationships or constraints.

Question:
Please list the location coordinates of the US circuits.

Output Rules:
1. You must output exactly one <think>...</think> block followed by one
   <answer>...</answer> block.
2. Inside <answer>, include exactly one ```sql ... ``` block containing
   runnable SQLite SQL.
3. Do not output any text before <think> or after </answer>.
4. Do not use <sql>...</sql>.
5. Keep <think> concise (<=120 words) so the final <answer> is never dropped.

Evidence:
location coordinates refers to (lat, lng); the US refers to country = 'USA';

Instructions:
- Make sure you only output the information that is asked in the question. If
  the question asks for a specific column, make sure to only include that column
  in the SELECT clause, nothing more.
- The generated query should return all of the information asked in the question
  without any missing or extra information.
- Note that while the reasoning process and SQL query need to be enclosed within
  <think> </think> and <answer> </answer> tags respectively, this should not
  affect the quality of the SQL generation.
- The answer must contain the SQL query within ```sql ``` tags.

Output Format:
In your answer, please enclose the generated SQL query in a code block:
"""sql
-- Your SQL query
"""

Take a deep breath and think step by step to find the correct SQL query.
\end{tcblisting}


\newpage
\section*{NeurIPS Paper Checklist}

\begin{enumerate}

\item {\bf Claims}
    \item[] Question: Do the main claims made in the abstract and introduction accurately reflect the paper's contributions and scope?
    \item[] Answer: \answerYes{} 
    \item[] Justification: We claimed in the abstract and introduction that we proposed an automatic annotation pipeline for clause-level reward annotation, developed a lightweight clause-level process reward model, improved Text-to-SQL execution accuracy, and improved failure localization.

\item {\bf Limitations}
    \item[] Question: Does the paper discuss the limitations of the work performed by the authors?
    \item[] Answer: \answerYes{} 
    \item[] Justification: We discuss limitations in \Cref{sec:limitations}, including dataset and SQL-dialect scope, clause-granularity limits, and the computational overhead of annotation.

\item {\bf Theory assumptions and proofs}
    \item[] Question: For each theoretical result, does the paper provide the full set of assumptions and a complete (and correct) proof?
    \item[] Answer: \answerYes{} 
    \item[] Justification: We provide the lemmas and corresponding proofs in \Cref{app:reward_propagation_theory}.

    \item {\bf Experimental result reproducibility}
    \item[] Question: Does the paper fully disclose all the information needed to reproduce the main experimental results of the paper to the extent that it affects the main claims and/or conclusions of the paper (regardless of whether the code and data are provided or not)?
    \item[] Answer: \answerYes{} 
    \item[] Justification: The paper provides detailed descriptions of the dataset curation process, model and parameter configuration, benchmarks and evaluation metrics, baselines, and the environment used for experiments in the main text. Additional details on data splits, hyperparameters, and other experimental settings are provided in the supplemental material.

\item {\bf Open access to data and code}
    \item[] Question: Does the paper provide open access to the data and code, with sufficient instructions to faithfully reproduce the main experimental results, as described in supplemental material?
    \item[] Answer: \answerYes{} 
    \item[] Justification: We provide anonymized supplemental materials for data construction, model training, and evaluation, including scripts and instructions for reproducing the main experimental results.

\item {\bf Experimental setting/details}
    \item[] Question: Does the paper specify all the training and test details (e.g., data splits, hyperparameters, how they were chosen, type of optimizer) necessary to understand the results?
    \item[] Answer: \answerYes{} 
    \item[] Justification: The paper specifies the dataset curation process, model and parameter configuration, benchmarks, evaluation metrics, baselines, and experimental environment in the main text and appendix. Additional data splits, hyperparameters, and inference details are provided in \Cref{app:exp_details}.

\item {\bf Experiment statistical significance}
    \item[] Question: Does the paper report error bars suitably and correctly defined or other appropriate information about the statistical significance of the experiments?
    \item[] Answer: \answerNo{} 
    \item[] Justification: We do not report error bars, confidence intervals, or statistical significance tests because full repeated RL training with large language models is computationally expensive. We instead report controlled comparisons, ablations, and results across multiple datasets and decoding settings.

\item {\bf Experiments compute resources}
    \item[] Question: For each experiment, does the paper provide sufficient information on the computer resources (type of compute workers, memory, time of execution) needed to reproduce the experiments?
    \item[] Answer: \answerYes{} 
    \item[] Justification: The paper reports the GPU type and memory, the RL framework, key training hyperparameters, and measured wall-clock time and GPU-hours for Clause-PRM training, GRPO policy optimization, and evaluation in \Cref{app:exp_details}.
    
\item {\bf Code of ethics}
    \item[] Question: Does the research conducted in the paper conform, in every respect, with the NeurIPS Code of Ethics \url{https://neurips.cc/public/EthicsGuidelines}?
    \item[] Answer: \answerYes{} 
    \item[] Justification: The work uses public Text-to-SQL benchmarks, model checkpoints, software, and API services under their stated licenses or access terms; does not involve human subjects, crowdsourcing, or newly collected personal data; and documents limitations, potential negative societal impacts, compute resources, release documentation, and asset licenses in the paper, appendix, checklist, and supplementary material.

\item {\bf Broader impacts}
    \item[] Question: Does the paper discuss both potential positive societal impacts and negative societal impacts of the work performed?
    \item[] Answer: \answerYes{} 
    \item[] Justification: The paper discusses positive impacts from improving the efficiency and accessibility of Text-to-SQL systems, as well as potential negative impacts from incorrect SQL generation, over-reliance on automated database access, and misuse on sensitive databases.
    
\item {\bf Safeguards}
    \item[] Question: Does the paper describe safeguards that have been put in place for responsible release of data or models that have a high risk for misuse (e.g., pre-trained language models, image generators, or scraped datasets)?
    \item[] Answer: \answerNA{} 
    \item[] Justification: The paper does not release data or models that have a high risk for misuse; the released artifacts are limited to Text-to-SQL annotation, training, and evaluation assets.

\item {\bf Licenses for existing assets}
    \item[] Question: Are the creators or original owners of assets (e.g., code, data, models), used in the paper, properly credited and are the license and terms of use explicitly mentioned and properly respected?
    \item[] Answer: \answerYes{} 
    \item[] Justification: We cite the original papers for datasets, models, and software, summarize their licenses or provider access terms in \Cref{tab:existing_assets}, and include asset links and release notes in the supplemental README.

\item {\bf New assets}
    \item[] Question: Are new assets introduced in the paper well documented and is the documentation provided alongside the assets?
    \item[] Answer: \answerYes{} 
    \item[] Justification: We release a new clause-annotated dataset and a clause-level process reward model; the README in the supplemental code package documents data construction, splits, training configuration, evaluation protocol, and limitations.

\item {\bf Crowdsourcing and research with human subjects}
    \item[] Question: For crowdsourcing experiments and research with human subjects, does the paper include the full text of instructions given to participants and screenshots, if applicable, as well as details about compensation (if any)? 
    \item[] Answer: \answerNA{} 
    \item[] Justification: The paper does not involve crowdsourcing or research with human subjects.

\item {\bf Institutional review board (IRB) approvals or equivalent for research with human subjects}
    \item[] Question: Does the paper describe potential risks incurred by study participants, whether such risks were disclosed to the subjects, and whether Institutional Review Board (IRB) approvals (or an equivalent approval/review based on the requirements of your country or institution) were obtained?
    \item[] Answer: \answerNA{} 
    \item[] Justification: The paper does not involve crowdsourcing or research with human subjects.

\item {\bf Declaration of LLM usage}
    \item[] Question: Does the paper describe the usage of LLMs if it is an important, original, or non-standard component of the core methods in this research? Note that if the LLM is used only for writing, editing, or formatting purposes and does \emph{not} impact the core methodology, scientific rigor, or originality of the research, declaration is not required.
    \item[] Answer: \answerYes{} 
    \item[] Justification: The paper describes the use of LLMs as the Text-to-SQL policy, clause-level process reward model, and evaluation baselines in the methodology and experimental setup. LLM assistance used only for writing and editing did not affect the core methodology, scientific rigor, or originality of the work.

\end{enumerate}

\end{document}